%% file: arxiv.tex
\documentclass[authorversion,screen,nonacm]{acmart}

\AtBeginDocument{%
   \providecommand\BibTeX{{%
         \normalfont B\kern-0.5em{\scshape i\kern-0.25em
            b}\kern-0.8em\TeX}}}

\setcopyright{none} 

\def\RegVerFlag{1}
\input{prefix}

\begin{document}


\title{Sampling a Near Neighbor in High Dimensions} \subtitle{--- Who
   is the Fairest of Them All?}

\regVer{
\author{Martin Aum\"uller} \affiliation{ \institution{IT University of
      Copenhagen, Denmark} } \email{maau@itu.dk}

\author{Sariel Har-Peled} \affiliation{ \institution{University of
      Illinois at Urbana-Champaign (UIUC)} \country{United States}}
\email{sariel@illinois.edu}

\author{Sepideh Mahabadi} \affiliation{ \institution{Toyota
      Technological Institute at Chicago (TTIC)} \country{United
      States}} \email{mahabadi@ttic.edu}

\author{Rasmus Pagh} \affiliation{\institution{BARC and University of
      Copenhagen, Denmark}} \email{pagh@di.ku.dk}

\author{Francesco Silvestri} \affiliation{\institution{University of
      Padova, Italy}} \email{silvestri@dei.unipd.it}
}

\dbVer{
    \author{Undisclosed Authors}
}

\regVer{
\renewcommand{\shortauthors}%
{M. Aum\"uller, S. Har-Peled, S. Mahabadi, R. Pagh, and F. Silvestri}
}


\begin{abstract}
    Similarity search is a fundamental algorithmic primitive, widely
    used in many computer science disciplines. Given a set of points
    $\PS$ and a radius parameter $r>0$, the $r$-near neighbor ($r$-NN)
    problem asks for a data structure that, given any query point $q$,
    returns a point $p$ within distance at most $r$ from $q$.  In this
    paper, we study the $r$-NN problem in the light of individual
    fairness and providing equal opportunities: all points that are
    within distance $r$ from the query should have the same
    probability to be returned. In the low-dimensional case, this
    problem was first studied by Hu, Qiao, and Tao (PODS 2014).
    Locality sensitive hashing (LSH), the theoretically strongest
    approach to similarity search in high dimensions, does not provide
    such a fairness guarantee.

    In this work, we show that \LSH based algorithms can be made fair,
    without a significant loss in efficiency. We propose several
    efficient data structures for the exact and approximate variants
    of the fair NN problem.  Our approach works more generally for
    sampling uniformly from a sub-collection of sets of a given
    collection and can be used in a few other applications.  We also
    develop a data structure for fair similarity search under inner
    product that requires nearly-linear space and exploits locality
    sensitive filters.  The paper concludes with an experimental
    evaluation that highlights the inherent unfairness of NN data
    structures and shows the performance of our algorithms on
    real-world datasets.

    \regVer{Preliminary versions of the results of this paper were published in
    \cite{hm-nnwft-19,aps-fnnsi-20}.}
\end{abstract}



\maketitle

\section{Introduction}
In recent years, following a growing concern about the fairness of the
algorithms and their bias toward a specific population or
feature~\cite{hps-eosl-16, c-fpdis-17, msp-bdras-16,
   kleinberg2017human}, there has been an increasing interest in
building algorithms that achieve (appropriately defined)
\emph{fairness} \cite{dwork2012fairness}.
The goal is to remove, or at least minimize, unethical behavior such
as discrimination and bias in algorithmic decision making, as
nowadays, many important decisions, such as college admissions,
offering home loans, or estimating the likelihood of recidivism, rely
on machine learning algorithms.  While algorithms are not inherently
biased, nevertheless, they may amplify the already existing biases in
the data.  Hence, this concern has led to the design of fair
algorithms for many different applications, e.g.,
\cite{donini2018empirical,abdlw-rafc-18,pleiss2017fairness,
   pmlr-v89-chierichetti19a,elzayn2019fair,olfat2018convex,
   chierichetti2017fair,backurs2019scalable,bera2019fair,
   kleindessner2019guarantees}.

There is no unique definition of fairness~(see \cite{hps-eosl-16} and
references therein), but different formulations that depend on the
computational problem at hand, and on the ethical goals we aim for.
Fairness goals are often defined in the political context of
socio-technical systems~\cite{Whi16}, and have to be seen in an
interdisciplinary spectrum covering many fields outside computer
science~\cite{Selbst19}.
In particular, researchers have studied both {\em group
   fairness}\footnote{The concept is denoted as statistical fairness
   too, e.g.,~\cite{c-fpdis-17}.}  (where demographics of the
population are preserved in the outcome), and {\em individual
   fairness} (where the goal is to treat individuals with similar
conditions similarly) \cite{dwork2012fairness}.  The latter concept of
``equal opportunity'' requires that people who can achieve a certain
advantaged outcome, such as finishing a university degree, or paying
back a loan, have equal opportunity of being able to get access to it
in the first place. %

Bias in the data used for training machine learning algorithms is a
monumental challenge in creating fair
algorithms~\cite{huang2007correcting, torralba2011unbiased,
   zafar2017fairness, c-fpdis-17}. Here, we are interested in a
somewhat different problem of handling the bias introduced by the data
structures used by such algorithms.
Specifically, data structures may introduce bias in the data stored in
them, and the way they answer queries, because of the way the data is
stored and how it is being accessed.  It is also possible that some
techniques for boosting performance, like randomization and
approximation that result in non-deterministic behavior, add to the
overall algorithmic bias. For instance, some database indexes for fast
search might give an (unexpected) advantage to some portions of the
input data.  Such a defect leads to selection bias by the algorithms
using such data structures.  It is thus natural to want data
structures that do not introduce a selection bias into the data when
handling queries.
To this end, imagine a data structure that can return, as an answer to
a query, an item out of a set of acceptable answers. The purpose is
then to return uniformly a random item out of the set of acceptable
outcomes, without explicitly computing the whole set of acceptable
answers (which might be prohibitively expensive).

\paragraph{The Near Neighbor Problem.} In this work, we study
similarity search and in particular the near neighbor problem from the
perspective of individual fairness.  Similarity search is an important
primitive in many applications in computer science such as machine
learning, recommender systems, data mining, computer vision, and many
others; see~\cite{sti-nnmlv-06, ai-nohaa-08} for an overview.  One of
the most common formulations of similarity search is the $r$-near
neighbor ($r$-NN) problem, formally defined as follows.  Let
$(\MS,\dist)$ be a metric space where the distance function
$\dist(\cdot, \cdot)$ reflects the (dis)similarity between two data
points. Given a set $\PS \subseteq \MS$ of $n$ points and a radius
parameter $r$, the goal of the $r$-NN problem is to preprocess $\PS$
and construct a data structure, such that for a query point
$q \in \MS$, one can report a point $p \in \PS$, such that
$\dist(p,q) \leq r$ if such a point exists.  As all the existing
algorithms for the \emph{exact} variant of the problem have either
space or query time that depends exponentially on the ambient
dimension of $\MS$, people have considered the approximate variant of
the problem. In the \emph{$c$-approximate near neighbor} (\ANN)
problem, the algorithm is allowed to report a point $p$ whose distance
to the query is at most $cr$ if a point within distance $r$ of the
query exists, for some prespecified constant $c > 1$.

\paragraph{Fair Near Neighbor.} As we will see, common existing data
structures for similarity search have a behavior that introduces bias
in the output. Our goal is to capture and algorithmically remove this
bias from these data structures.
Our goal is to develop a data structure for the $r$-near neighbor
problem where we aim to be fair among ``all the points'' in the
neighborhood, i.e., all points within distance $r$ from the given
query have the same probability to be returned. We introduce and study
the \emph{fair near neighbor} problem: if $\nbrY{\q}{r}$ is the ball
of input points at distance at most $r$ from a query $\q$, we would
like that each point in $\nbrY{\q}{r}$ is returned as near neighbor of
$\q$ with the uniform probability of $1/\nNY{\q}{r}$ where
$\nNY{\q}{r}=|\nbrY{\q}{r}|$.

\paragraph{Locality Sensitive Hashing.} Perhaps the most prominent
approach to get an \ANN data structure for high-dimensional data is
via the Locality Sensitive Hashing (\LSH) framework proposed by Indyk
and Motwani \cite{IndykM98,him-anntr-12}, which leads to sub-linear
query time and sub-quadratic space. In particular, for $\MS=\Re^d$, by
using \LSH one can get a query time of $n^{\rho+o(1)}$ and space
$n^{1+\rho+o(1)}$ where for the $L_1$ distance metric $\rho=1/c$
\cite{IndykM98,him-anntr-12}, and for the $L_2$ distance metric
$\rho=1/c^2 + o_c(1)$ \cite{ai-nohaa-08}.  In the \LSH framework,
which is formally introduced in \secref{lsh}, the idea is to
hash all points using several hash functions that are chosen randomly,
with the property that closer points have a higher probability of
collision than the far points.  Thus, the collision probability
between two points is a decreasing function of their
distance~\cite{Charikar02}. Therefore, the closer points to a query
have a higher probability of falling into a bucket being probed than
far points. Thus, reporting a random point from a random bucket
computed for the query, produces a distribution that is biased by the
distance to the query: closer points to the query tend to have a
higher probability of being chosen.  On the other hand, the uniformity
property required in fair \NN can be trivially achieved by finding
\emph{all} $r$-near neighbor of a query and then randomly selecting
one of them.  This is computationally inefficient since the query time
is a function of the size of the neighborhood.  One contribution in
this paper is the description of much more efficient data structures
that still use LSH in a black-box way.


\paragraph{Applications: When random nearby is better than nearest.}
The bias mentioned above towards nearer points is usually a good
property, but is not always desirable. Indeed, consider the following
scenarios: 
\medskip%
\begin{compactenumI}[leftmargin=0.8cm,itemsep=-0.5ex]
    \item The nearest neighbor might not be the best if the input is
    noisy, and the closest point might be viewed as an
    unrepresentative outlier. Any point in the neighborhood might be
    then considered to be equivalently beneficial. This is to some
    extent why $k$-\NN classification \cite{ell-ca-09} is so effective
    in reducing the effect of noise.
    
    \item However, $k$-\NN works better in many cases if $k$ is large,
    but computing the $k$ nearest-neighbors is quite expensive if $k$
    is large \cite{haaa-spkp-14}. Computing quickly a random nearby
    neighbor can significantly speed-up such classification.
    
    \item If one wants to estimate the number of items with a desired
    property within the neighborhood, then the easiest way to do it is
    via uniform random sampling from the neighborhood. In particular,
    this is useful for density estimation \cite{klk-oknnd-12}.  More
    generally, this can be seen as a special case of query sampling in
    database systems~\cite{Olken1995}, where the goal is to return a
    random sample of the output of a given query, for efficiently
    providing statistics on the query.  This can for example be used
    for estimating aggregate queries (e.g., \texttt{sum} or
    \texttt{count}), see~\cite{Olken1995Survey} for more details.
    Another example for the usefulness is discrimination discovery in
    existing databases~\cite{LuongRT11}: by performing independent
    queries to obtain a sample with statistical significance, we can
    reason about the distribution of attribute types. We could report
    on discrimination if the population counts grouped by a certain
    attribute differ much more than we would expect them to.

    \item We are interested in anonymizing the query
    \cite{a-u4aql-07}, thus returning a random near-neighbor might
    serve as the first line of defense in trying to make it harder to
    recover the query. Similarly, one might want to anonymize the
    nearest-neighbor \cite{qa-eppkn-08}, for applications were we are
    interested in a ``typical'' data item close to the query, without
    identifying the nearest item.

    \item As another application, consider a recommender system used
    by a newspaper to recommend articles to users.  Popular
    recommender systems based on matrix factorization~\cite{KorenBV09}
    give recommendations by computing the inner product similarity of
    a user feature vector with all item feature vectors using some
    efficient similarity search algorithm.  It is common practice to
    recommend those items that have the largest inner product with the
    user.  However, in general it is not clear that it is desirable to
    recommend the ``closest'' articles.  Indeed, it might be desirable
    to recommend articles that are on the same topic but are not
    \emph{too} aligned with the user feature vector, and may provide a
    different perspective~\cite{Abiteboul17}.  As described by
    Adomavicius and Kwon in~\cite{adomavicius2014optimization},
    recommendations can be made more diverse by sampling $k$ items
    from a larger top-$\ell$ list of recommendations at random. Our
    data structures could replace the final near neighbor search
    routine employed in such systems.

    \item Another natural application is simulating a random walk in
    the graph where two items are connected if they are in distance at
    most $r$ from each other. Such random walks are used by some graph
    clustering algorithms \cite{hk-curw-01}.
\end{compactenumI}

\smallskip To the best of our knowledge, previous results focused on
exact near neighbor sampling in the Euclidean space up to three
dimensions \cite{ap-irsra-19,aw-irsr-17,hqt-irs-14,Olken1995}.
Although these results might be extended to $\mathbb{R}^d$ for any
$d>1$, they suffer from the \emph{curse of dimensionality} as the
query time increases exponentially with the dimension, making the data
structures too expensive in high dimensions.  These bounds are
unlikely to be significantly improved since several conditional lower
bounds show that an exponential dependency on $d$ in query time or
space is unavoidable for \emph{exact} near neighbor search (see, e.g.,
\cite{AlmanR15,Williams05}).

\subsection{Problem formulations}
\seclab{problems}

In the following we formally define the variants of the fair \NN
problem that we consider in this paper.  For all constructions
presented, these guarantees hold only in the absence of a failure
event that happens with probability at most $\delta$ for some
arbitrarily small $\delta > 0$.

\begin{definition}[$r$-near neighbor sampling, i.e., {Fair NN with
    dependence}]
    \deflab{nns}%
    Consider a set $\PS\subseteq \X$ of $n$ points in a metric space
    $(\X, \D)$.  The \emph{$r$-near neighbor sampling problem}
    ($r$-\NNS) asks to construct a data structure for $\PS$ to solve
    the following task with probability at least $1 - \delta$: Given
    query $\q$, return a point $\p$ uniformly sampled from the set
    $\nbrY{\q}{r}$. We also refer to this problem as Fair \NN with
    Dependence.
\end{definition}

Observe that the definition above does not require different query
results to be independent.  If the query algorithm is deterministic
and randomness is only used in the construction of the data structure,
the returned near neighbor of a query will always be the same.
Furthermore, the result of a query $\q$ might be correlated with the
result of a different query $\q'$.  This motivates us to extend the
$r$-\NNS problem to the scenario where we aim at independence.

\begin{definition}[$r$-near neighbor independent sampling, i.e., {Fair
    \NN}] \deflab{nnis} Consider a set $\PS\subseteq \X$ of $n$
    points in a metric space $(\X, \D)$.  The \emph{$r$-near neighbor
       independent sampling problem} ($r$-\NNIS) asks to construct a
    data structure for~$\PS$ that for any sequence of up to $n$
    queries $\q_1, \q_2,\ldots, \q_n$ satisfies the following
    properties with probability at least $1 - \delta$:
    \begin{enumerate}
        \item For each query $\q_i$, it returns a point
        $\textnormal{OUT}_{i,\q_i}$ uniformly sampled from
        $\nbrY{\q_i}{r}$;
        \item The point returned for query $\q_i$, with $i>1$, is
        independent of previous query results. That is, for any
        $\p\in \nbrY{\q_i}{r}$ and any sequence
        $\p_1,\ldots,\p_{i-1}$, we have that
        \begin{equation*}
            \Pr[ \textnormal{OUT}_{i,\q} {=} \p \mid \textnormal{OUT}_{i-1,\q_{i - 1}} {=} \p_{i - 1}, \ldots, \textnormal{OUT}_{1,\q_1} {=} \p_1] = \frac{1}{\nNY{\q}{r}}.
        \end{equation*}
    \end{enumerate}
    We also refer to this problem as \emph{Fair \NN}.
\end{definition}
We note that in the low-dimensional
setting~\cite{hqt-irs-14,aw-irsr-17,ap-irsra-19}, the $r$-near
neighbor independent sampling problem is usually called
\emph{independent range sampling} (IRS).  Next, motivated by
applications, we define two approximate variants of the problem that
we study in this work.  More precisely, we slightly relax the fairness
constraint, allowing the probabilities of reporting a neighbor to be
an ``almost uniform" distribution.

\begin{definition}[{\AFNN}]
    Consider a set $\PS\subseteq \X$ of $n$ points in a metric space
    $(\X, \D)$.  The \emph{\AFNN} problem asks to construct a data
    structure for~$\PS$ that for any query $\q$, returns each point
    $\p\in \nbrY{\q}{r}$ with probability $\mu_\p$ where $\mu$ is an
    approximately uniform probability distribution:
    $\pNY{\q}{r}/ (1+\eps) \leq \mu_\p \leq (1+\eps)\pNY{\q}{r}$,
    where $\pNY{\q}{r} = 1/\nNY{\q}{r}$.  We again assume the same
    independence assumption as in \defref{nnis}.
\end{definition}

Next, we allow the algorithm to report an almost uniform distribution
from an \emph{approximate} neighborhood of the query.

\begin{definition}[{\AFANN}]
    Consider a set $\PS\subseteq \X$ of $n$ points in a metric space
    $(\X, \D)$.  The \emph{\AFANN} problem asks to construct a data
    structure for~$\PS$ that for any query $\q$, returns each point
    $\p\in S'$ with probability $\mu_\p$ where
    $\prb/(1+\eps) \leq \mu_\p \leq (1+\eps)\prb$, where $S'$ is a
    point set such that
    $\nbrY{\q}{r}\subseteq S' \subseteq \nbrY{\q}{cr}$, and
    $\prb = 1/|S'|$.  We again assume the same independence assumption
    as in \defref{nnis}.
\end{definition}

\subsection{Our results}
\seclab{results} We propose several solutions to the different
variants of the Fair \NN problem.  Our solutions make use of the
\LSH framework~\cite{IndykM98} and we denote by $\dsS(n,c)$ the space
usage and by $\dsQ(n,c)$ the running time of a standard \LSH data
structure that solves the $c$-\ANN problem in the space $(\X, \D)$.
\begin{itemize}
    \item \secref{fair:nn:dependent} describes a solution to
    the Fair \NN problem with dependence with expected running time
    $\tldO(\dsQ(n,c) + n(\q, cr) - n(\q, r))$ and space
    $\dsS(n,c) + O(n)$.  The data structure uses an independent
    permutation of the data points on top of a standard LSH data
    structure and inspects buckets according to the order of points
    under this permutation.  See \thmref{fair:nn:dependent} for
    the exact statement.
    \item In \secref{approx:fair:nn} we provide a data structure for
    \AFANN that uses space $\dsS(n,c)$ and whose query time is
    $\tldO(\dsQ(n,c))$, both in expectation and also with high
    probability (using slightly different bounds). See \thmref{approx-neighborhood} for the exact
    statement.
    \item \secref{independent:fair:nn} shows how to solve the
    Fair \NN problem in expected query time
    $\tldO(\dsQ(n,c) + n(\q, cr) / n(\q, r))$ and space usage
    $O(\dsS(n,c))$.  Each bucket is equipped with a count-sketch and
    the algorithm works by repeatedly sampling points within a certain
    window from the permutation.  See
    \thmref{fair:nn:independent} for the exact statement.
    \item In \secref{tableau} we introduce an easy-to-implement
    nearly-linear space data structure based on the locality-sensitive
    filter approach put forward
    in~\cite{AndoniLRW17,christiani2017framework}.  As each input
    point appears once in the data structure, the data structure can
    be easily adapted to solve the Fair \NN problem.
    While conceptually simple, it does not use \LSH as a black-box
    and works only for some distances: we describe it for similarity
    search under inner product, although it can be adapted to some
    other metrics (like Euclidean and Hamming distances) with standard
    techniques.  See \thmref{tableau:fair:nn} for the exact
    statement.
    \item Lastly, in \secref{evaluation} we present an
    empirical evaluation of (u{}n)fairness in traditional recommendation
    systems on real-world datasets, and we then analyze the additional
    computational cost for guaranteeing fairness.  More precisely, we
    compare the performance of our algorithms with the algorithm that
    uniformly picks a bucket and reports a random point, on five
    different datasets using both Euclidean distance and Jaccard
    similarity.  Our empirical results show that while the standard
    \LSH algorithm fails to fairly sample a point in the neighborhood
    of the query, our algorithms produce empirical distributions which
    are much closer to the uniform distribution.  We further include a
    case study highlighting the unfairness that might arise in special
    situations when considering \AFANN.
\end{itemize}

We remark that for the approximate variants, the dependence of our
algorithms on $\eps$ is only $O(\log (1/\eps))$.  While we omitted the
exact poly-logarithmic factors in the list above, they are generally
lower for the approximate versions.  Furthermore, these methods can
be embedded in the existing LSH method to achieve unbiased query
results in a straightforward way.  On the other hand, the exact
methods will have higher logarithmic factors and use additional data
structures.

\subsection{Data structure for sampling from a sub-collection of sets}

In order to obtain our results, we first study a more generic problem
in \secref{union-of-sets}: given a collection $\Family$ of
sets from a universe of $n$ elements, a query is a sub-collection
$\FamilyA\subseteq \Family$ of these sets and the goal is to sample
(almost) uniformly from the union of the sets in this sub-collection.
We also show how to modify the data structure to handle outliers in
\secref{outliers}, as it is the case for \LSH, as the
sampling algorithm needs to ignore such points once they are reported
as a sample.  This will allow us to derive most of the results
concerning variants of Fair \NN in \secref{s:f:nn} as
corollaries from these more abstract data structures.

\bigskip\noindent%
\textbf{Applications.}  Here are a few examples of applications of a
data structure that provides uniform samples from a union of sets:
\smallskip%
\begin{compactenumA}
    \item Given a subset $\setA$ of vertices in the graph, randomly
    pick (with uniform distribution) a neighbor to one of the vertices
    of $\setA$. This can be used in simulating disease spread
    \cite{ke-nem-05}.  \smallskip
    \item Here, we use a variants of the data structure to implement
    Fair \NN.

    \smallskip%
    \item Uniform sampling for range searching \cite{hqt-irs-14,
       aw-irsr-17, ap-irsra-19}. Indeed, consider a set of points,
    stored in a data structure for range queries. Using the above, we
    can support sampling from the points reported by several queries,
    even if the reported answers are not disjoint.
\end{compactenumA}
\smallskip%
Being unaware of any previous work on this problem, we believe this
data structure is of independent interest.

\subsection{Discussion of Fairness Assumptions}
In the context of our problem definition we assume---as do many papers
on fairness-related topics---an implicit world-view described by
Friedler \etal \cite{FriedlerSV16} as ``what you see is what you
get''. WYSIWYG means that a certain distance between individuals in
the so-called ``construct space'' (the true merit of individuals) is
approximately represented by the feature vectors in ``observed
space''. As described in their paper, one has to subscribe to this
world-view to achieve certain fairness conditions. Moreover, we
acknowledge that our problem definition requires to set a threshold
parameter $r$ which might be internal to the dataset. This problem
occurs frequently in the machine learning community, e.g., when score
thresholding is applied to obtain a classification result. Kannan
\etal discuss the fairness implications of such threshold approaches
in~\cite{Kannan}.

We stress that the $r$-near neighbor independent sampling problem
might not be {\bf\emph{the}} fairness definition in the context of
similarity search.  Instead, we think of it as a suitable starting
point for discussion, and acknowledge that the application will often
motivate a suitable fairness property.  For example, in the case of a
recommender system, we might want to consider a weighted case where
closer points are more likely to be returned. As discussed earlier,
and exemplified in the experimental evaluation, a standard LSH
approach does not have such guarantees despite its monotonic collision
probability function.  We leave the weighted case as an interesting
direction for future work.

\dbVer{
\subsection{Related Work}

The different problem variants discussed in \secref{problems} were 
first introduced in~\cite{nipsanon} (for approximate fair near neighbor search) 
and~\cite{podsanon} (for fair near neighbor search).
Moreover, \cite{nipsanon} also introduced the more abstract problem of sampling uniformly from a sub-collection of sets. 
Both papers present algorithms that achieve running time bounds that are roughly similar to the main statements presented in \secref{results}. 
In this paper, we extend the work in \cite{podsanon} by describing variants that work in more abstract setting of sampling from a sub-collection.
Furthermore, we discuss high probability bound on the running time of the query algorithm in \secref{independent:fair:nn}.
This papers simplifies the algorithms to solve approximate fair near neighbors that were described in \cite{nipsanon}. 
Lastly, this paper presents a unified experimental view on the problem of sampling a fair near neighbor. 
}


\section{Preliminaries}
\seclab{prelims}

\textbf{Set representation.} %
Let $\ground$ be an underlying ground set of $n$ objects (i.e.,
elements). In this paper, we deal with sets of objects. Assume that
such a set $\setA\subseteq \ground$ is stored in some reasonable data
structure, where one can insert delete, or query an object in constant
time. Querying for an object $\obj\in \ground$, requires deciding if
$\obj \in \setA$. Such a representation of a set is straightforward to
implement using an array to store the objects, and a hash table.  This
representation allows random access to the elements in the set, or
uniform sampling from the set.

If hashing is not feasible, one can just use a standard dictionary
data structure -- this would slow down the operations by a logarithmic
factor.

\medskip

\noindent
\textbf{Subset size estimation.}  We need the following standard
estimation tool, \cite[Lemma 2.8]{bhrrs-eeiso-17}.

\begin{lemma}
    \lemlab{est:set}%
    Consider two sets $\SC \subseteq \SA$, where $n = \cardin{\SA}$.
    Let $\epsA, \BadProb \in (0,1)$ be parameters, such that
    $\BadProb < 1/ \log n$. Assume that one is given an access to a
    membership oracle that, given an element $x \in \SA$, returns
    whether or not $x \in \SC$. Then, one can compute an estimate $s$,
    such that
    $(1-\epsA)\cardin{\SC}\leq s \leq (1+\epsA)\cardin{\SC}$, and
    computing this estimates requires
    $O( (n/\cardin{\SC}) \epsA^{-2} \log \BadProb^{-1})$ oracle
    queries. The returned estimate is correct with probability
    $\geq 1 - \BadProb$.
\end{lemma}

\noindent
\textbf{Weighted sampling.}  We need the following standard data
structure for weighted sampling.

\begin{lemma}\lemlab{ds-tree}
    Given a set of objects $\FamilyB = \brc{ \obj_1, \ldots, \obj_t}$,
    with associated weights $w_1,\ldots, w_t$, one can preprocess them
    in $O(t)$ time, such that one can sample an object out of
    $\FamilyB$.  The probability of an object $\obj_i$ to be sampled
    is $w_i / \sum_{j=1}^t w_j$. In addition the data structure
    supports updates to the weights. An update or sample operation
    takes $O( \log t)$ time.
\end{lemma}

\begin{proof}
    Build a balanced binary tree $\Tree$, where the objects of
    $\FamilyA$ are stored in the leaves.  Every internal node $u$ of
    $\Tree$, also maintains the total weight $w(u)$ of the objects in
    its subtree. The tree $\Tree$ has height $O( \log t)$, and weight
    updates can be carried out in $O( \log t)$ time, by updating the
    path from the root to the leaf storing the relevant object.

    Sampling is now done as follows -- we start the traversal from the
    root. At each stage, when being at node $u$, the algorithm
    considers the two children $u_1,u_2$. It continues to $u_1$ with
    probability $w(u_1)/ w(u)$, and otherwise it continues into
    $u_2$. The object sampled is the one in the leaf that this
    traversal ends up at.
\end{proof}%

\noindent
\textbf{Sketch for distinct elements.}\seclab{sketch} In
\secref{union-of-sets-independent} we will use sketches for
estimating the number of distinct elements.  Consider a stream of $m$
elements $x_1,\ldots, x_m$ in the domain $[n] = \{1,\ldots, n\}$ and
let $F_0$ be the number of distinct elements in the stream (i.e., the
zeroth-frequency moment).  Several papers have studied sketches (i.e.,
compact data structures) for estimating $F_0$.  For the sake of
simplicity we use the simple sketch in \cite{BarYossefJKST02}, which
generalizes the seminal result by Flajolet and Martin
\cite{FlajoletM85}.  The data structure consists of
$\Delta=\BT{\log(1/\delta)}$ lists $L_1,\ldots L_\Delta$; for
$1 \leq w \leq \Delta$, $L_w$ contains the $t = \BT{1/\epsilon^2}$
distinct smallest values of the set
$\{\psi_w(x_i) : 1\leq i \leq m\}$, where
$\psi_w : [n] \rightarrow [n^3]$ is a hash function picked from a
pairwise independent family.  It is shown in \cite{BarYossefJKST02}
that the median $\hat F_0$ of the values
$tn^3/v_0, \ldots, tn^3/v_\Delta$, where $v_w$ denotes the $t$\th
smallest value in $L_w$, is an $\epsilon$-approximation to the number
of distinct elements in the stream with probability at least
$1-\delta$: that is,
$(1-\epsilon) F_0 \leq \hat F_0 \leq (1+\epsilon) F_0$.  The data
structure requires $\BO{\epsilon^{-2} \log m \log (1/\delta)}$ bits
and $\BO{\log(1/\epsilon) \log m \log (1/\delta)}$ query time.  A nice
property of this sketch is that if we split the stream in $k$ segments
and we compute the sketch of each segment, then it is possible to
reconstruct the sketch of the entire stream by combining the sketches
of individual segments (the cost is linear in the sketch size).

\section{Data structure for sampling from the union of sets}
\seclab{union-of-sets}

\noindent
\textbf{The problem.}  Assume you are given a data structure that
contains a large collection $\Family$ of sets of objects.  In total,
there are $n = \cardin{\bigcup \Family}$ objects.  The sets in
$\Family$ are not necessarily disjoint. The task is to preprocess the
data structure, such that given a sub-collection
$\FamilyA \subseteq \Family$ of the sets, one can quickly pick
uniformly at random an object from the set
\begin{math}
    {\textstyle \bigcup} \FamilyA%
    :=%
    \bigcup_{\setA \in \FamilyA} \setA.
\end{math}

\noindent
\textbf{Naive solution.}  The naive solution is to take the sets under
consideration (in $\FamilyA$), compute their union, and sample
directly from the union set ${\textstyle \bigcup} \FamilyA$. Our
purpose is to do (much) better -- in particular, the goal is to get a
query time that depends logarithmically on the total size of all sets
in $\FamilyA$.

\noindent
\textbf{Parameters.} The query is a family
$\FamilyA \subseteq \Family$, and define
$m = \Cardin{\FamilyA} := \sum_{\setA \in \FamilyA} \cardin{\setA}$
(which should be distinguished from $\nSets = \cardin{\FamilyA}$ and
from $N= \cardin{\bigcup \FamilyA}$).

\noindent
\textbf{Preprocessing.}
For each set $\setA \in \Family$, we build the set representation
mentioned in the preliminaries section.  In addition, we assume that
each set is stored in a data structure that enables easy random access
or uniform sampling on this set (for example, store each set in its
own array). Thus, for each set $\setA$, and an element, we can decide
if the element is in $\setA$ in constant time.

\noindent
\textbf{Variants.}
In the
same way as there were multiple fairness definitions in
\secref{problems}, we can wish for a one-shot sample from
$\bigcup \FamilyA$ (allowing for \emph{dependence}) or for
\emph{independent} results. Moreover, sample probabilities can be
\emph{exact} or \emph{approximate}.  Since all elements are valid
elements to return, there is no difference between an \emph{exact} or
\emph{approximate} neighborhood.  This will be the topic of
\secref{outliers}.

\noindent
\textbf{Outline.}
The approaches discussed in the following use two ideas: (i) using a
random permutation of the universe to introduce a natural order of the
elements and (ii) using rejection sampling to introduce randomness
during the query to guarantee an (approximately) equal output
probability.

The first approach in \secref{union-of-sets-dependent} uses
only the random permutation and results in one-shot uniform sample,
lacking independence.  The data structures in
\secref{uniform}--\secref{almost:uniform} build on
top of rejection sampling. They provide independence, but introduce
approximate probabilities.  Finally, the data structure in
\secref{union-of-sets-independent} makes use of both ideas to
produce an independent sample with exact probabilities.

\subsection{Uniform sampling with dependence}
\seclab{union-of-sets-dependent}
We start with a simple data structure for sampling a uniform point
from a collection of sets, i.e., given a sub-collection $\FamilyA$,
sample a point in $\bigcup \FamilyA$ uniformly at random.  Since all
randomness is in the preprocessing of the data structure, this variant
does not guarantee independence regarding multiple queries.

The main idea is quite simple.  We initially assign a (random) rank to
each of the $n$ objects in $\bigcup \Family$ using a random
permutation. We sort all sets in $\Family$ according to the ranks of
their objects.  For a given query collection $\FamilyA$, we iterate
through each set and keep track of the element with minimum rank in
$\bigcup \FamilyA$.  This element will be returned as answer to the
query.  The random permutation guarantees that all points in
$\FamilyA$ have the same chance of being returned.

\begin{lemma}
    \lemlab{set:dependent}
    Let $N = \cardin{\bigcup \FamilyA}$, $g = \cardin{\FamilyA}$, and
    $m = \sum_{X \in \FamilyA} \cardin{X}$.  The above algorithm
    samples an element $x \in \bigcup \FamilyA$ according to the
    uniform distribution in time $O(g)$.
\end{lemma}

\begin{proof}
    The algorithm keeps track of the element of minimum rank among the
    $g$ different sets.  Since all of them are sorted during
    preprocessing, this takes time $O(1)$ per set.  The sample is
    uniformly distributed because each element has the same chance of
    being the smallest under the random permutation.
\end{proof}

If we repeat the same query on the same sub-collection $\FamilyA$, the
above data structure always returns the same point: if we let $OUT_i$
denote the output of the $i$\th sample from $\FamilyA$, we have that
$\PR{OUT_i = x | OUT_1=x_1}$ is 1 if $x=x_1$ and 0 otherwise.  We now
extend the above data structure to get independent samples when we
repeat the query on the \emph{same} sub-collection $\FamilyA$, that is
\begin{align*}
  \PR{OUT_i =x | OUT_{i-1}=x_{i-1}, \ldots OUT_{1}=x_{1}}= \PR{OUT_i = x} = 1/N
\end{align*}
We add to each set in $\Family$ a priority queue which supports key
updates, using ranks as key.  For each point $x\in \bigcup \Family$,
we keep a pointer to all sets (and their respective priority queue)
containing $x$.  At query time, we search the point in sets $\FamilyA$
with minimum rank, as in the previous approach.  Then, just before
returning the sample, we apply a small random perturbation to ranks
for ``destroying'' any relevant information that can be collected by
repeating the query.  The perturbation is obtained by applying a
random swap, similar to the one in the Fisher-Yates shuffle
\cite{Knuth97}: let $r_x$ be the rank of $x$; we randomly select a
rank $r$ in $\{r_k,\ldots n\}$ and let $y$ be the point with rank $r$;
then we swap the ranks of $x$ and $y$ and update accordingly the
priority queues.  We have the following lemma, where $\delta$ denotes
the maximum number of sets in $\Family$ containing a point in
$\bigcup \Family$.

\begin{lemma}
    \lemlab{rep_query_dependent}
    Let $n = \cardin{\bigcup \Family}$,
    $N = \cardin{\bigcup \FamilyA}$, $g = \cardin{\FamilyA}$, and
    $\delta=\max_{x\in \bigcup \Family} |\{A \in \Family | x \in A\}|$
    Assume to repeat $k$ times the above query procedure on the
    sub-collection $\FamilyA$, and let $OUT_i$ denote the output of
    the $i$\th iteration with $1 \leq i \leq k$.  Then, for any
    $x\in \bigcup\FamilyA$, we have that $\PR{OUT_1 = x}=1/N$ and
    $\PR{OUT_i = x | OUT_{i-1}=x_{i-1}, \ldots OUT_{1}=x_{1}} =
    \PR{OUT_1 = x} = 1/N$ for any $i>1$.  Each query requires
    $\BO{(g+\delta) \log n}$ time.
\end{lemma}
\begin{proof}
    Let $L$ be the set of points in $\bigcup\Family$ with ranks larger
    than $r_x$. We have that $|L|=n-r_x$ and
    $\bigcup\FamilyA\setminus\{x\} \subseteq L$. Before the swap, the
    ranks of points in $L$ are unknown and all permutations of points
    in $L$ are equally likely.  After the swap, each point in
    $L\bigcup \{x\}$ has rank $r_x$ with probability $1/(n-r_x+1)$,
    independently of previous random choices.  Moreover, each point in
    $\bigcup\FamilyA$ has probability $1/N$ to be the point in
    $\bigcup\FamilyA$ with smaller rank after the swap.

    By assuming that each priority queue can be updated in
    $\BO{\log n}$ time, the point with minimum rank can be extracted in $\BO{g \log n}$ time, and the final rank shuffle requires
    $\BO{\delta \log n}$ time as we need to update the priority queues
    of the at most $2\delta$ sets containing point $x$ or point $y$.
\end{proof}

We remark that the re{}randomization technique is only restricted to
single element queries: over time all elements in $\FamilyA$ get
higher and higher ranks: this means that for another collection
$\FamilyA'$, which intersects $\FamilyA$, the elements in
$\bigcup \FamilyA \setminus \bigcup \FamilyA'$ become more and more
likely to be returned. The next sections provide slightly more
involved data structures that guarantee independence even among
different queries.

\subsection{Uniform sampling via exact degree computation}
\seclab{uniform}

The query is a family $\FamilyA \subseteq \Family$.  The
\emphi{degree} of an element $x \in \bigcup \FamilyA$, is the number
of sets of $\FamilyA$ that contains it -- that is,
$\degY{\FamilyA}{x} = \cardin{\DegY{\FamilyA}{x}}$, where
$ \DegY{\FamilyA}{x}%
=%
\Set{ \setA \in \FamilyA}{ x \in \setA }.  $
The algorithm repeatedly does the following:
\begin{compactenumI}
    \regVer{\smallskip}%
    \item \itemlab{s:sample}%
    Picks one set from $\FamilyA$ with probabilities proportional to
    their sizes. That is, a set $\setA \in \FamilyA$ is picked with
    probability $\cardin{\setA} / m$.
    
    \item \itemlab{b:sample}%
    It picks an element $x \in \setA$ uniformly at random.
    
    \item Computes the degree $\degC = \degY{\FamilyA}{x}$.
    
    \item Outputs $x$ and stop with probability $1/\degC$. Otherwise,
    continues to the next iteration.
\end{compactenumI}

\begin{lemma}
    \lemlab{q:exact}%
    Let $N = \cardin{\bigcup\FamilyA}$, $\nSets = \cardin{\FamilyA}$,
    and $m = \sum_{X \in \FamilyA} \cardin{X}$. The above algorithm
    samples an element $x \in \bigcup \FamilyA$ according to the
    uniform distribution. The algorithm takes in expectation
    $O( \nSets m/N ) = O( \nSets^2 )$ time. The query time is
    $O(\nSets^2 \log N)$ with high probability.
\end{lemma}

\begin{proof}
    Observe that an element $x \in \bigcup \FamilyA$ is picked by step
    \itemref{b:sample} with probability $\alpha = \degX{x}/m$. The
    element $x$ is output with probability $\beta = 1/ \degX{x}$. As
    such, the probability of $x$ to be output by the algorithm in this
    round is $\alpha \beta = 1/ \Cardin{\FamilyA}$. This implies that
    the output distribution is uniform on all the elements of
    $\bigcup \FamilyA$.
    
    The probability of success in a round is $N/m$, which implies that
    in expectation $m/N$ rounds are used, and with high probability
    $O((m/N) \log N)$ rounds. Computing the degree
    $\degY{\FamilyA}{x}$ takes $O( \cardin{\FamilyA})$ time, which
    implies the first bound on the running time. As for the second
    bound, observe that an element can appear only once in each set of
    $\FamilyA$, which readily implies that
    $\degX{y} \leq \cardin{\FamilyA}$, for all
    $y \in \bigcup \FamilyA$.
\end{proof}%

\subsection{Almost uniform sampling via degree approximation}

The bottleneck in the above algorithm is computing the degree of an
element. We replace this by an approximation.

\begin{definition}
    Given two positive real numbers $x$ and $y$, and a parameter
    $\eps \in (0,1)$, the numbers $x$ and $y$ are
    \emphi{$\eps$-approximation} of each other, denoted by
    $x \aprxEps y$, if $x/(1+\eps) \leq y \leq x(1+\eps)$ and
    $y/(1+\eps) \leq x \leq
    y(1+\eps)$.
\end{definition}
In the approximate version, given an item $x \in \bigcup \FamilyA$, we
can approximate its degree and get an improved runtime for the
algorithm.

\begin{lemma}%
    \lemlab{almost-uniform}%
    The input is a family of sets $\Family$ that one can preprocess in
    linear time.  Let $\FamilyA\subseteq\Family$ be a sub-family and
    let $N = \cardin{\bigcup\FamilyA}$, $\nSets = \cardin{\FamilyA}$,
    and $\eps \in (0,1)$ be a parameter.  One can sample an element
    $x \in \bigcup \FamilyA$ with almost uniform probability
    distribution.  Specifically, the probability of an element to be
    output is $\aprxEps 1/N$. After linear time preprocessing, the
    query time is $O\pth{ \nSets \eps^{-2} \log N}$, in expectation,
    and the query succeeds with high probability.
\end{lemma}
\begin{proof}
    Let $m = \Cardin{\FamilyA}$.  Since
    $\degX{x} = \cardin{\DegY{\FamilyA}{x}}$, it follows that we need
    to approximate the size of $\DegY{\FamilyA}{x}$ in
    $\FamilyA$. Given a set $\setA \in \FamilyA$, we can in constant
    time check if $x \in \setA$, and as such decide if
    $\setA \in \DegY{\FamilyA}{x}$. It follows that we can apply the
    algorithm of \lemref{est:set}, which requires
    \begin{math}
        W(x)%
        =%
        O\bigl( \tfrac{\nSets}{ \degX{x}} \eps^{-2} \log N\bigr)
    \end{math}
    time, where the algorithm succeeds with high probability. The
    query algorithm is the same as before, except that it uses the
    estimated degree.

    For $x \in \bigcup \FamilyA$, let $\Event_x$ be the event that the
    element $x$ is picked for estimation in a round, and let
    $\Event_x'$ be the event that it was actually output in that
    round.  Clearly, we have $\ProbCond{\Event_x' }{\Event_x} = 1/d$,
    where $d$ is the degree estimate of $x$. Since
    $d \aprxEps \degX{x}$ (with high probability), it follows that
    $\ProbCond{\Event_x' }{\Event_x} \aprxEps 1/\degX{x}$. Since there
    are $\degX{x}$ copies of $x$ in $\FamilyA$, and the element for
    estimation is picked uniformly from the sets of $\FamilyA$, it
    follows that the probability of any element
    $x \in \bigcup \FamilyA$ to be output in a round is
    \begin{equation*}
        \Prob{\Event_x'}%
        = 
        \ProbCond{\Event_x'}{\Event_x}
        \Prob{\Event_x}
        = 
        \ProbCond{\Event_x'}{\Event_x}
        \frac{ \degX{x}}{ m}  \aprxEps 1/m,
    \end{equation*}
    as $\Event_x' \subseteq \Event_x$.  As such, the probability of
    the algorithm terminating in a round is
    \begin{math}
        \alpha %
        = %
        \sum_{x \in \bigcup \FamilyA} \Prob{\Event_x'}%
        \aprxEps N/m%
        \geq%
        N/2m.
    \end{math}
    As for the expected amount of work in each round, observe that it
    is proportional to
    \begin{equation*}
        W%
        =%
        \sum_{x \in \bigcup \FamilyA} \Prob{\Event_x} W(x)
        =%
        \sum_{x \in \bigcup \FamilyA}
        \frac{\degX{x}}{m}
        \frac{ \nSets }{ \eps^2 \degX{x}} \log N
        =%
        O\pth{ \frac{ n \nSets}{m} \eps^{-2} \log N}.
    \end{equation*}
    
    Intuitively, since the expected amount of work in each iteration
    is $W$, and the expected number of rounds is $1/\alpha$, the
    expected running time is $O( W / \alpha)$. This argument is not
    quite right, as the amount of work in each round effects the
    probability of the algorithm to terminate in the round (i.e., the
    two variables are not independent). We continue with a bit more
    care -- let $L_i$ be the running time in the $i$\th round of the
    algorithm if it was to do an $i$\th iteration (i.e., think about a
    version of the algorithm that skips the experiment in the end of
    the iteration to decide whether it is going to stop), and let
    $Y_i$ be a random variable that is $1$ if the (original) algorithm
    had not stopped at the end of the first $i$ iterations of the
    algorithm.
    
    By the above, we have that
    $y_i = \Prob{Y_{i}=1} = \ProbCond{Y_{i}=1}{Y_{i-1}
       =1}\Prob{Y_{i-1}=1} \leq (1-\alpha)y_{i-1} \leq (1-\alpha)^i$,
    and $\Ex{L_i} = O(W)$. Importantly, $L_i$ and $Y_{i-1}$ are
    independent (while $L_i$ and $Y_i$ are dependent). We clearly have
    that the running time of the algorithm is
    $O \bigl( \sum_{i=1}^\infty Y_{i-1}L_i \bigr)$ (here, we define
    $Y_0 =1$). Thus, the expected running time of the algorithm is
    proportional to
    \begin{align*}
      \Ex{ \Bigl. \smash{\sum_{i} Y_{i-1}L_i }}%
      &=%
        \sum_{i} \Ex{Y_{i-1}L_i }
        =%
        \sum_{i} \Ex{Y_{i-1}} \Ex{L_i }
        \leq%
        W \sum_{i} y_{i-1}
        \leq%
        W \sum_{i=1}^\infty (1-\alpha)^{i-1}%
        =%
        \frac{W}{\alpha}\\
      &=%
        O( \nSets \eps^{-2} \log N),
    \end{align*}
    because of linearity of expectations, and since $L_i$ and
    $Y_{i-1}$ are independent.
\end{proof}

\begin{remark}\remlab{whp:approx}
    The query time of \lemref{almost-uniform} deteriorates to
    $O\pth{ \nSets \eps^{-2} \log^2 T}$ if one wants the bound to hold
    with high probability, where $T$ is some (rough) upper bound on
    $N$. This follows by restarting the query algorithm if the query
    time exceeds (say by a factor of two) the expected running time.
    A standard application of Markov's inequality implies that this
    process would have to be restarted at most $O(\log T)$ times, with
    high probability. Here, one can set $T$ to be $n \cdot \nSets $ as
    a rough upper bound on $N$.
\end{remark}

\begin{remark}
    The sampling algorithm is independent of whether or not we fully
    know the underlying family $\Family$ and the sub-family
    $\FamilyA$. This means the past queries do not affect the sampled
    object reported for the query $\FamilyA$. Therefore, the almost
    uniform distribution property holds in the presence of several
    queries and independently for each of them.
\end{remark}

\subsection{Almost uniform sampling via simulation}
\seclab{almost:uniform}

It turns out that one can avoid the degree approximation stage in the
above algorithm, and achieve only a polylogarithmic dependence on
$\eps^{-1}$.  To this end, let $x$ be the element picked. We need to
simulate a process that accepts $x$ with probability $1/\degX{x}$.

We start with the following natural idea for estimating $\degX{x}$ --
probe the sets randomly (with replacement), and stop in the $i$\th
iteration if it is the first iteration where the probe found a set
that contains $x$. If there are $\nSets$ sets, then the distribution
of $i$ is geometric, with probability $p = \degX{x}/\nSets$. In
particular, in expectation, $\Ex{i} = \nSets/\degX{x}$, which implies
that $\degX{x} = \nSets /\Ex{i}$. As such, it is natural to take
$\nSets/i$ as an estimation for the degree of $x$. Thus, to simulate a
process that succeeds with probability $1/\degX{x}$, it would be
natural to return $1$ with probability $i/\nSets$ and $0$
otherwise. Surprisingly, while this seems like a heuristic, it does
work, under the right interpretation, as testified by the following.%

\begin{lemma}
    \lemlab{first}%
    Assume we have $\nSets$ urns, and exactly $\degC > 0$ of them, are
    non-empty. Furthermore, assume that we can check if a specific urn
    is empty in constant time. Then, there is a randomized algorithm,
    that outputs a number $Y \geq 0$, such that $\Ex{Y}=1/\degC$. The
    expected running time of the algorithm is $O(\nSets/\degC)$.
\end{lemma}%
\begin{proof}
    The algorithm repeatedly probes urns (uniformly at random), until
    it finds a non-empty urn. Assume it found a non-empty urn in the
    $i$\th probe. The algorithm outputs the value $i/\nSets$ and
    stops.
    
    Setting $p = \degC/\nSets$, and let $Y$ be the output of the
    algorithm. we have that
    \begin{equation*}
        \Ex{\bigl. Y}%
        =%
        \sum_{i=1}^\infty \frac{i}{\nSets} (1-p)^{i-1} p
        =%
        \frac{p}{\nSets(1-p)}\sum_{i=1}^\infty i (1-p)^{i} 
        =%
        \frac{p}{\nSets(1-p)} \cdot \frac{1-p}{p^2}
        =%
        \frac{1}{p\nSets}%
        =%
        \frac{1}{\degC},
    \end{equation*}
    using the formula $\sum_{i=1}^\infty ix^i = {x}/{(1-x)^2}$.

    The expected number of probes performed by the algorithm until it
    finds a non-empty urn is $1/p = \nSets/\degC$, which implies that
    the expected running time of the algorithm is $O(\nSets/\degC)$.~
\end{proof}%

The natural way to deploy \lemref{first}, is to run its algorithm to
get a number $y$, and then return $1$ with probability $y$. The
problem is that $y$ can be strictly larger than $1$, which is
meaningless for probabilities. Instead, we backoff by using the value
$y/\Delta$, for some parameter $\Delta$. If the returned value is
larger than $1$, we just treat it at zero.  If the zeroing never
happened, the algorithm would return one with probability
$1/(\degX{x}\Delta)$ -- which we can use to our purposes via,
essentially, amplification. Instead, the probability of success is
going to be slightly smaller, but fortunately, the loss can be made
arbitrarily small by taking $\Delta$ to be sufficiently large.

\begin{lemma}
    \lemlab{second}%
    There are $\nSets$ urns, and exactly $\degC > 0$ of them are not
    empty. Furthermore, assume one can check if a specific urn is
    empty in constant time. Let $\BadProb \in (0,1)$ be a
    parameter. Then one can output a number $Z \geq 0$, such that
    $Z \in [0,1]$, and
    \begin{math}
        \Ex{Z} \in I= \bigl[ \tfrac{1}{\degC\Delta} - \BadProb,
        \tfrac{1}{\degC\Delta} \bigr],\Bigr.%
    \end{math}
    where
    $\Delta = \ceil{\smash{\ln \BadProb^{-1}} } + 4 = \Theta(\log
    \BadProb^{-1})$. The expected running time of the algorithm is
    $O( \nSets/\degC)$.

    Alternatively, the algorithm can output a bit $X$, such that
    $\Prob{X=1} \in I$.
\end{lemma}
\begin{proof}
    We modify the algorithm of \lemref{first}, so that it outputs
    $i/(\nSets\Delta)$ instead of $i/\nSets$. If the algorithm does
    not stop in the first $\nSets\Delta+1$ iterations, then the
    algorithm stops and outputs $0$. Observe that the probability that
    the algorithm fails to stop in the first $\nSets\Delta$
    iterations, for $p = \degC / \nSets$, is
    \begin{math}
        (1-p)^{\nSets\Delta} \leq \exp\pth{ -\frac{\degC}{\nSets}
           \nSets \Delta}%
        \leq%
        \exp( -\degC \Delta) \leq%
        \exp( - \Delta) \ll \BadProb.
    \end{math}

    Let $Z$ be the random variable that is the number output by the
    algorithm. Arguing as in \lemref{first}, we have that
    \begin{math}
        \Ex{Z} \leq 1/(\degC\Delta).
    \end{math}
    More precisely, we have
    \begin{math}
        \Ex{Z}%
        =%
        \frac{1}{\degC\Delta} - \sum_{i=\nSets\Delta+1}^\infty
        \frac{i}{\nSets \Delta} (1-p)^{i-1} p.
    \end{math}
    Let
    \begin{align*}
      \sum_{i=\nSets j+1}^{\nSets(j+1)} \frac{i}{\nSets} (1-p)^{i-1} p
      &\leq%
        (j+1)\sum_{i=\nSets j+1}^{\nSets(j+1)}  (1-p)^{i-1} p
        =%
        (j+1)(1-p)^{\nSets j}\sum_{i=0}^{\nSets - 1}  (1-p)^{i} p
      \\
      &\leq%
        (j+1)(1-p)^{\nSets j}
        \leq%
        (j+1) \pth{1-\frac{\degC}{\nSets}}^{\nSets j}
        \leq%
        (j+1) \exp \pth{- \degC j }.
    \end{align*}

    Let $g(j) = \frac{j+1}{\Delta} \exp \pth{- \degC j }$.  We have
    that
    \begin{math}
        \Ex{Z}%
        \geq%
        \frac{1}{\degC \Delta} - \beta,
    \end{math}
    where $\beta = \sum_{j=\Delta}^\infty g(j)$.  Furthermore, for
    $j \geq \Delta$, we have
    \begin{equation*}
        \frac{g(j+1)}{g(j)} %
        =%
        \frac
        {  (j+2) \exp \pth{- \degC (j+1) }}
        {  (j+1) \exp \pth{- \degC j }}%
        \leq %
        \pth{1+\frac{1}{\Delta}} e^{-\degC}
        \leq %
        \frac{5}{4} e^{-\degC}
        \leq%
        \frac{1}{2}.
    \end{equation*}
    As such, we have that
    \begin{equation*}
        \beta%
        =%
        \sum_{j=\Delta}^\infty g(j)%
        \leq%
        2 g(\Delta)%
        \leq%
        2 \frac{\Delta+1}{\Delta} \exp \pth{- \degC \Delta }%
        \leq 
        4 \exp \pth{- \Delta }%
        \leq%
        \BadProb,
    \end{equation*}
    by the choice of value for $\Delta$. This implies that
    $\Ex{Z} \geq 1/(\degC\Delta) - \beta \geq 1/(\degC \Delta) -
    \BadProb$, as desired.

    The alternative algorithm takes the output $Z$, and returns $1$
    with probability $Z$, and zero otherwise.
\end{proof}%

\begin{lemma}%
    \lemlab{almost:uniform:2}%
    The input is a family of sets $\Family$ that one preprocesses in
    linear time.  Let $\FamilyA\subseteq\Family$ be a sub-family and
    let $N = \cardin{\bigcup\FamilyA}$, $\nSets = \cardin{\FamilyA}$,
    and let $\eps \in (0,1)$ be a parameter.  One can sample an
    element $x \in \bigcup \FamilyA$ with almost uniform probability
    distribution.  Specifically, the probability of an element to be
    output is $\aprxEps 1/N$. After linear time preprocessing, the
    query time is $O\pth{ \nSets \log (\nSets/\eps)}$, in expectation,
    and the query succeeds, with probability
    $\geq 1 - 1/\nSets^{O(1)}$.
\end{lemma}
\begin{proof}
    The algorithm repeatedly samples an element $x$ using steps
    \itemref{s:sample} and \itemref{b:sample} of the algorithm of
    \secref{uniform}. The algorithm returns $x$ if the algorithm of
    \lemref{second}, invoked with $\BadProb = (\eps/\nSets)^{O(1)}$
    returns $1$. We have that $\Delta = \Theta( \log(\nSets/\eps) )$.
    Let $\alpha = 1/(\degX{x}\Delta)$.  The algorithm returns $x$ in
    this iteration with probability $p$, where
    $p \in [\alpha - \BadProb, \alpha]$.  Observe that
    $\alpha \geq 1/(\nSets\Delta)$, which implies that
    $\BadProb \ll (\eps/4 )\alpha$, it follows that
    $p \aprxEps 1/(\degX{x}\Delta)$, as desired.  The expected running
    time of each round is $O(\nSets/\degX{x})$.

    Arguing as in \lemref{almost-uniform}, this implies that each
    round, in expectation takes $O\pth{ N \nSets / m }$ time, where
    $m = \Cardin{\FamilyA}$. Similarly, the expected number of rounds,
    in expectation, is $O(\Delta m/N)$. Again, arguing as in
    \lemref{almost-uniform}, implies that the expected running time is
    $O(\nSets \Delta ) =O( \nSets \log (\nSets/\eps))$.
\end{proof}%

\begin{remark}%
    \remlab{whp:simul}%
    Similar to \remref{whp:approx}, the query time of
    \lemref{almost:uniform:2} can be made to work with high
    probability with an additional logarithmic factor. Thus with high
    probability, the query time is
    $O\pth{ \nSets \log(\nSets/\eps) \log N}$.
\end{remark}

\subsection{Uniform sampling using random ranks}
\seclab{union-of-sets-independent}

In this section, we present a data structure that samples an element
uniformly at random from $\bigcup \FamilyA$ using both ideas from the
previous subsections: we assign a random rank to each object as in
\secref{union-of-sets-dependent}, and use rejection sampling
to provide independent and uniform output probabilities.  Let
$\Lambda$ be the sequence of the $n = \cardin{\bigcup \Family}$ input
elements after a random permutation; the rank of an element is its
position in $\Lambda$.  We first highlight the main idea of the query
procedure.

Let $k\geq 1$ be a suitable value that depends on the collection
$\FamilyA$ and assume that $\Lambda$ is split into $k$ segments
$\Lambda_i$, with $i\in\{0,\ldots, k-1\}$. (We assume for simplicity
that $n$ and $k$ are powers of two.)  Each segment $\Lambda_i$
contains the $n/k$ elements in $\Lambda$ with rank in
$[i\cdot n/ k, (i+1)\cdot n/ k)$.  We denote with
$\lambda_{\FamilyA,i}$ the number of elements from $\bigcup \FamilyA$
in $\Lambda_i$, and with $\lambda \geq \max_i\{\lambda_{\FamilyA,i}\}$
an upper bound on the number of these elements in each segment.  By
the initial random permutation, we have that each segment contains at
most $\lambda=\BT{(N/k)\log n}$ elements from $\bigcup \FamilyA$ with
probability at least $1-1/n^2$.  (Of course, $N$ is \emph{not} known
at query time.)

The query algorithm works in the following three steps in which all
random choices are independent.

\begin{compactenumI}
    \item \itemlab{uniform:s:sample}%
    Select uniformly at random an integer $h$ in $\{0,\ldots, k-1\}$
    (i.e., select a segment $\Lambda_h$);
        
    \item \itemlab{uniform:reject}%
    With probability $\lambda_{\FamilyA,h}/\lambda$ move to
    step~\itemref{uniform:output}, otherwise repeat
    step~\itemref{uniform:s:sample};
    
    \item \itemlab{uniform:output} Return an element uniformly sampled
    among the elements in $\bigcup \FamilyA$ in $\Lambda_h$.
    
\end{compactenumI}

Since each object in $\bigcup \FamilyA$ has a probability of
$1/(k\lambda)$ of being returned in \stepref{uniform:output},
the result is a uniform sample of $\bigcup \FamilyA$.  The algorithm
described above works for all choices of $k$, but a good choice has to
depend on $\FamilyA$ for the following reasons.  On the one hand, the
segments should be small, because otherwise
Step~\itemref{uniform:output} will take too long.  On the other hand,
they have to contain at least one element from $\bigcup \FamilyA$,
otherwise we sample many ``empty'' segments in
Step~\itemref{uniform:s:sample}.  We will see that the number $k$ of
segments should be roughly set to $N$ to balance the trade-off.
However, the number $N$ of distinct elements in $\bigcup \FamilyA$ is
not known.  Thus, we set $k=2\hat{s_{\FamilyA}}$, where
$\hat{s_\FamilyA}$ is a $1/2$-approximation of $N$.  Such an estimate
can be computed by storing a count distinct sketch for each set in
$\Family$.  To compute $\hat{s_\FamilyA}$ we merge the count distinct
sketches of all $g$ sets of $\FamilyA$.  
To compute $\lambda_{\FamilyA, h}$ efficiently, we assume that, at construction time, the elements in each set in
$\Family$ are sorted by their rank.

\begin{lemma}
    \lemlab{union-of-sets-independent}
    Let $N = \cardin{\bigcup \FamilyA}$, $g = \cardin{\FamilyA}$,
    $m = \sum_{X \in \FamilyA} \cardin{X}$, and
    $n = \cardin{\bigcup \Family}$.  With probability at least
    $1-1/n^2$, the algorithm described above returns an element
    $x \in \bigcup \FamilyA$ according to the uniform distribution.
    The algorithm has an expected running time of $O(g\log^2 n)$.
\end{lemma}

\begin{proof}
    We start by bounding the initial failure probability of the data
    structure.  By a union bound, we have that the following two
    events hold simultaneously with probability at least $1-1/n^2$:
    \begin{enumerate}
        \item Count distinct sketches provide a $1/2$-approximation of
        $N$.  By setting $\delta=1/(2n^2)$ in the count distinct
        sketch construction (see \secref{sketch}), the
        approximation guarantee holds with probability at least
        $1-1/(2n^2)$.
        \item Every segment of size $n/k$ contains no more than
        $\lambda=\BT{\log n}$ elements from $\bigcup \FamilyA$.  As
        elements are initially randomly permuted, the claim holds with
        probability at least $1-1/(2n^2)$ by suitably setting the
        constant in $\lambda=\BT{\log n}$.
    \end{enumerate}

    From now on assume that these events are true.

    Each element in $\bigcup \FamilyA$ has the same probability
    $1/(k\Lambda)$ of being returned in
    \stepref{uniform:output}, so all points are equally likely
    to be sampled.
    Note also that the guarantees are independent of the initial
    random permutation as soon as the two events above hold.  This
    means that the data structure returns a uniform sample from a
    union-of-sets.

    We now focus on the time complexity of the query algorithm.  In
    \stepref{uniform:reject}, $\lambda_{\FamilyA, h}$ is
    computed by iterating through the $g$ sets and collection points
    using a range query on segment $\Lambda_h$.  
    Since elements in each set are sorted by their rank, the range query can be carried out by searching for rank $h n/k$ using a binary search in $O(\log n)$ time, and then enumerating all elements with rank smaller than $(h + 1)n/k$.
    This takes time $O(\log n + o)$ for each set, where $o$ is the output size.
    Since
    each segment contains $O(\log n)$ elements from
    $\bigcup \FamilyA$, one iteration of
    \stepref{uniform:reject} takes time $O(g \log n)$.

    By our choice of $k$, we have that
    $\lambda_{\FamilyA,h}/\lambda = \Theta(1/\log n)$, thus we expect
    to carry out $\Theta(\log n)$ iterations before
    reaching~\itemref{uniform:output}.  Thus, we expect to spend time
    $O(g \log^2 n)$ before reaching \stepref{uniform:output}.
    \stepref{uniform:output} again takes time $O(g \log n)$, with the
    same analysis as above.
\end{proof}

\begin{remark}\remlab{whp:union:of:sets:independent}
    The query time of \lemref{union-of-sets-independent} can be
    made to work with high probability with an additional logarithmic
    factor. Thus with high probability, the query time is
    $O\pth{ \nSets \log^3(n)}$.
\end{remark}

\section{Handling outliers}
\seclab{outliers}

Imagine a situation where we have a marked set of outliers $\OL$. We
are interested in sampling from $\bigcup \FamilyA \setminus \OL$.  We
assume that the total degree of the outliers in the query is at most
$\mOL$ for some prespecified parameter $\mOL$. More precisely, we have
$\degY{\FamilyA}{\OL} = \sum_{x \in \OL} \degY{\FamilyA}{x} \leq
\mOL$.

\subsection{Sampling with Dependence}

We run a variant of the original algorithm from
\secref{union-of-sets-dependent}.  We use a priority queue
\texttt{P{}Q} to keep track of the point with smallest rank in
$\bigcup \FamilyA$.  Initially, for each $G \in \FamilyA$ we add the
pair $(x, G)$ to the priority queue, where $x$ is the element with the
smallest rank in $G$.  As long as the element with the smallest rank
in \PQ is not in $\bigcup \FamilyA \setminus \mathcal{O}$, we iterate
the following: Let $(x, G)$ be the entry extracted by an
\texttt{extract{}Min} operation on \PQ.  Let $y$ be the element in $G$
with the next largest rank.  Insert $(y, G)$ into \PQ. 

\begin{lemma}
\lemlab{outlierdep}
    The input is a family of sets $\Family$ that one can preprocess in
    linear time, and a query is a sub-family $\FamilyA \subseteq \Family$ and 
    a set of outliers $\OL$.  Let
    $N = \vert\bigcup \FamilyA \setminus \OL\vert$ and $g=\vert \FamilyA\vert$.  The above approach samples uniformly at random an element  $x \in \bigcup \FamilyA \setminus \OL$.
    The expected query time is
    $O(\log g (g + \degY{\FamilyA}{\OL} / (N + 1)))$, and it is never
    worse than $O(\log g (g + \degY{\FamilyA}{\OL}))$.
\end{lemma}

\begin{proof}
    For each $o \in \mathcal{O}$ and each $1 \leq i \leq g$, define
    the random variable $X_{o, i}$ that is 1 if $o$ is present in the
    $i$\th collection of $\FamilyA$ and has a rank smaller than all
    elements $\bigcup \FamilyA \setminus \mathcal{O}$.  By the initial
    random permutation, the probability that an outlier
    $o \in \mathcal{O}$ has a smaller rank than the $N$ elements in
    $\bigcup \FamilyA \setminus \mathcal{O}$ is exactly $1/(N + 1)$.
    Let $R$ be the number of rounds carried out by the query
    algorithm.  By linearity of expectation, we get:
    \begin{align*}
      \Ex{R} 
      & \leq g + \Ex{\sum_{o\in \mathcal{O}} \sum_{ i=1}^{g} X_{o,i}} = g + \frac{\degY{\FamilyA}{\OL}}{N + 1}. 
    \end{align*}
    The lemma follows because each round takes time $O(\log g)$ for
    the priority queue operations.
    Since an outlier cannot be sample twice, the algorithm stops after $\degY{\FamilyA}{\OL}$ rounds in the worst case, with  query time is
    $O((g + \degY{\FamilyA}{\OL}) \log g)$.
\end{proof}

Similarly to Lemma \lemref{rep_query_dependent}, we can extend the above data structure to support output independence if the same query $\FamilyA$ is repeated several times. 
It suffices to repeat the process until a point in $\bigcup \FamilyA \setminus \OL$ is found, and to apply the swap to the returned point.
Note that to efficiently perform swaps each set in $\Family$ should store points in a priority queue with ranks as keys. 
We get the following lemma.

\begin{lemma}\lemlab{ouliersrep}
    The input is a family of sets $\Family$ that one can preprocess in
    linear time. A query is a sub-family $\FamilyA \subseteq \Family$,
    and a set of outliers $\OL$.  Let
    $n = \cardin{\bigcup \Family}$, $N = \vert\bigcup \FamilyA \setminus \OL\vert$, $g=\vert \FamilyA\vert$ and
    $\delta=\max_{x\in \bigcup \Family} |\{A \in \Family | x \in A\}|$.
    Assume to repeat $k$ times the above query procedure on the sub-collection $\FamilyA$, and let $OUT_i$ denote the output of the $i$\th
    iteration with $1 \leq i \leq k$.  Then,  we have that $OUT_i\in \bigcup \FamilyA\setminus \OL$ for any $1\leq i\leq k$ and, for any $x\in \bigcup \FamilyA \setminus \OL$, 
    $\PR{OUT_1 = x}=1/N$ and
    $\PR{OUT_i = x | OUT_{i-1}=x_{i-1}, \ldots OUT_{1}=x_{1}} =
    \PR{OUT_1 = x} = 1/N$ for $i>1$.  The expected query
    time is $\BO{d_\FamilyA(\OL) \log n +  (g+\delta) \log n}$ time.
\end{lemma}
\begin{proof}
Initially, we need $\BO{g \log n}$ time to find the point in $\bigcup \FamilyA$ with smaller rank.
Then we need to repeat the procedure $\degY{\FamilyA}{\OL}/{(N + 1)}$ times in expectation since  the probability that an outlier
$o \in \mathcal{O}$ has a smaller rank than the $N$ elements in
 $\bigcup \FamilyA \setminus \mathcal{O}$ is  $1/(N + 1)$.
Since each repetition costs $\BO{\log n}$ and the final swap takes $\BO{\delta \log n}$ time, the expected running time follows.
The probabilities follows from \lemref{rep_query_dependent}.
\end{proof}

\subsection{Almost uniform sampling with outliers}

\begin{defn}
    \deflab{eps:uniform}%
    For a set $\PSA$, and a parameter $\eps \in [0,1)$, a sampling
    algorithm that outputs a sample $x \in \PSA$ generates
    \emphi{$\eps$-uniform distribution}, if for any $y \in \PSA$, we
    have that
    \begin{math}
        \displaystyle%
        \frac{1}{(1+\eps) \cardin{\PSA}}%
        \leq%
        \Prob{x =y}
        \leq%
        \frac{1+\eps}{\cardin{\PSA}}.
    \end{math}    
\end{defn}

\begin{lemma}
    \lemlab{outliers}%
    The input is a family of sets $\Family$ that one can preprocess in
    linear time. A query is a sub-family $\FamilyA \subseteq \Family$,
    a set of outliers $\OL$, a parameter $\mOL$, and a parameter
    $\eps \in (0,1)$.  One can either
    \begin{compactenumA}
        \item Sample an element $x \in \bigcup \FamilyA \setminus \OL$
        with $\eps$-uniform distribution.
        \item Alternatively, report that
        $\degY{\FamilyA}{\OL} > \mOL$.
    \end{compactenumA}
    The expected query time is
    $O\bigl( \mOL + \nSets \log (\nSets/\eps) \bigr)$, and the query
    succeeds, with probability $\geq 1-1/\nSets^{O(1)}$, where
    $\nSets = \cardin{\FamilyA}$.
\end{lemma}

\begin{proof}
    The main modification of the algorithm of
    \lemref{almost:uniform:2} is that whenever we encounter an outlier
    (the assumption is that one can check if an element is an outlier
    in constant time), then we delete it from the set $\setA$ where it
    was discovered. If we implement sets as arrays, this can be done
    by moving an outlier object to the end of the active prefix of the
    array, and decreasing the count of the active array. We also need
    to decrease the (active) size of the set. If the algorithm
    encounters more than $\mOL$ outliers then it stops and reports
    that the number of outliers is too large.
    
    Otherwise, the algorithm continues as before. The only difference
    is that once the query process is done, the active count (i.e.,
    size) of each set needs to be restored to its original size, as is
    the size of the set. This clearly can be done in time proportional
    to the query time.
\end{proof}%

\subsection{Uniform sampling with outliers}
\seclab{uniform:sampling:outliers}

We run a variant of the original algorithm from
\secref{union-of-sets-independent}.  In the same way as
before, we use the count distinct sketches to obtain an upper bound
$\hat{s_\FamilyA}$ on the number of distinct elements in $\FamilyA$.
Because of the presence of outliers, this bound will not necessarily
be close to $N$, but could be much larger.  Thus, we run the algorithm
at most $\log n$ rounds to find a suitable value of $k$.  In round
$i$, we use the value $k_i = 2\hat{s_\FamilyA}/2^i$.  Moreover, a
single round is iterated for $\Sigma = \Theta(\log^2 n)$ steps.  If
$k < 2$, we report that $\bigcup \FamilyA \setminus \mathcal{O}$ is
empty.  The precise algorithm is presented in the following. As
before, it takes an integer parameter $\mOL$ controlling the number of
outliers.

\begin{compactenumA}
    \item Merge all count distinct sketches of the $g$ sets in
    $\FamilyA$, and compute a $1/2$-approximation $\hat{s_\FamilyA}$
    of $s_\FamilyA = |\bigcup \FamilyA|$, such that
    $s_\FamilyA/2\leq \hat{s_\FamilyA} \leq 1.5 s_\FamilyA$.
    \item Set $k$ to the smallest power of two larger than or equal to
    $2 \hat{s_\FamilyA}$; let $\lambda = \BT{\log n}$,
    $\sigma_{\textnormal{fail}} = 0$ and $\Sigma=\BT{\log^2 n}$.
    \item Repeat the following steps until successful or $k<2$:
    \begin{compactenumI}
        \item Assume the input sequence $\Lambda$ to be split into $k$
        segments $\Lambda_i$ of size $n/k$, where $\Lambda_i$ contains
        the points in $\bigcup \Family \setminus \mathcal{O}$ with
        ranks in $[i \cdot n/k, (i+1)\cdot n/k)$. Denote with
        $\lambda_i$ the size of $\Lambda_i$.
        \item \itemlab{outlier:uniform:s:sample}%
        Select uniformly at random an integer $h$ in
        $\{0,\ldots, k-1\}$ (i.e., select a segment $\Lambda_h$);
        \item Increment $\sigma_{\textnormal{fail}}$.  If
        $\sigma_{\textnormal{fail}}=\Sigma$, then set $k = k/2$ and
        $\sigma_{\textnormal{fail}}=0$.
        \item \itemlab{outlier:uniform:reject}%
        Compute $\lambda_{\FamilyA, h}$ and count the number of
        outliers inspected on the way. If there are more than $\mOL$
        outliers, report that $\degY{\FamilyA}{\OL} > \mOL$.
        Otherwise, with probability $\lambda_{\FamilyA, h}/\lambda$,
        declare success.
    \end{compactenumI}
    \item If the previous loop ended with success, return an element
    uniformly sampled among the elements in
    $\bigcup \FamilyA \setminus \mathcal{O}$ in $\Lambda_h$, otherwise
    return $\perp$.
\end{compactenumA}

\begin{lemma}
    \lemlab{union-of-sets-independent:outliers}%
    The input is a family of sets $\Family$ that one can preprocess in
    linear time. A query is a sub-family $\FamilyA \subseteq \Family$,
    a set of outliers $\OL$, and a parameter $\mOL$.  With high
    probability, one can either:
    \begin{compactenumA}
        \item Sample a uniform element
        $x \in \bigcup \FamilyA \setminus \OL$, or
        \item Report that $\degY{\FamilyA}{\OL} > \mOL$.
    \end{compactenumA}
    The expected time is $O((g + \mOL)\log^4 n)$, where
    $N = \vert\bigcup \FamilyA \setminus \OL\vert$,
    $\nSets = \cardin{\FamilyA}$, and $n = \Cardin{\Family}$.
\end{lemma}

The proof will follow along the lines of the proof of
\lemref{union-of-sets-independent}.  We provide a
self-contained version for completeness and to highlight the
challenges of introducing outliers.

\begin{proof}
    We start by bounding the initial failure probability of the data
    structure.  By a union bound, we have that the following two
    events hold simultaneously with probability at least $1-1/n^2$:
    \begin{enumerate}
        \item Count distinct sketches provide a $1/2$-approximation of
        $|\bigcup \FamilyA|$.  By setting $\delta=1/(4n^2)$ in the
        count distinct sketch construction (see
        \secref{sketch}), the approximation guarantee holds
        with probability at least $1-1/(4n^2)$.
        \item When $k = 2^{\lceil{\log N}\rceil}$, every segment of
        size $n/k$ contains no more than $\lambda=\BT{\log N}$ points
        from $\bigcup \FamilyA \setminus \mathcal{O}$.  As points are
        initially randomly permuted, the claim holds with probability
        at least $1-1/(4n^2)$ by suitably setting the constant in
        $\lambda=\BT{\log n}$.
    \end{enumerate}

    From now on assume that these events are true.
    
    We will first discuss the additional failure event: $N \geq 1$,
    but the algorithm reports $\perp$.  The probability of this event
    is upper bounded by the probability $p'$ that no element is
    returned in the $\Sigma$ iterations where
    $k=2^{\lceil\log N\rceil}$ (the actual probability is even lower,
    since an element can be returned in an iteration where
    $k>2^{\lceil\log N\rceil}$).  By suitably setting constants in
    $\lambda = \BT{\log n}$ and $\Sigma=\BT{\log^2 n}$, we get:
    \begin{equation*}
        p' = \left( 1 - \frac{N}{k\lambda}\right)^{\Sigma} \leq
        e^{-\Sigma N / (k\lambda)} \leq e^{\BT{-\Sigma/\log n}} \leq \frac{1}{2n^2}.
    \end{equation*}

    By a union bound, with probability at least $1-1/n^2$ none of
    these three events are true.  To show that the returned element is
    uniformly sampled in $\bigcup \FamilyA \setminus \mathcal{O}$,
    recall that each element in
    $\bigcup \FamilyA \setminus \mathcal{O}$ has the same probability
    of $1/(k\lambda)$ of being output.

    For the running time, first focus on the round where
    $k = 2^{\lceil \log N\rceil}$.  In this round, we carry out
    $\Theta(\log^2 n)$ iterations.  In each iteration, we extract the
    points with rank in $[hn/k, (h + 1)n/k]$ from each of the $g$
    sets, counting all outlier points that we retrieve on the way.
    For each set, we expect to find $N/k = O(1)$ points in
    $\bigcup \FamilyA \setminus \mathcal{O}$.  If we retrieve more
    than $\mOL$ outliers, we report that
    $\degY{\FamilyA}{\OL} > \mOL$.  Reporting points with a given rank
    costs $O(\log n + o)$ in each bucket (where $o$ is the output
    size).  Thus, one iteration is expected to take time
    $O((g + \mOL)\log n)$.  The expected running time of all
    $\Sigma = \Theta(\log^2 n)$ iterations is bounded by
    $O((g +\mOL)\log^3 n)$.  Observe that for all the rounds carried
    out before, $k$ is only larger and thus the segments are smaller.
    This means that we may multiply our upper bound with $\log n$,
    which completes the proof.
\end{proof}

\section{In the search for a fair near neighbor}
\seclab{s:f:nn} In this section, we employ the data structures
developed in the previous sections to show the results on fair near
neighbor search listed in \secref{results}.

First, let us briefly give some preliminaries on \LSH. We refer the
reader to \cite{him-anntr-12} for further details. Throughout the
section, we assume that our metric space $(\X, \D)$ admits an \LSH
data structure.


\subsection{Background on \LSH}
\seclab{lsh}

Locality Sensitive Hashing (\LSH) is a common tool for solving the
\ANN problem and was introduced in \cite{IndykM98}.
\begin{definition}
    \deflab{lsh}
    A distribution $\mathcal H$ over maps $h\colon \X \rightarrow U$,
    for a suitable set $U$, is called
    $(r, c\cdot r, p_1, p_2)$-sensitive if the following holds for any
    $\x$, $\y\in \X$:
    \begin{itemize}
        \item if $\D(\x, \y) \leq r$, then
        $\Pr_h[h(\x) = h(\y)] \geq p_1$;
        \item if $\D(\x, \y) > c\cdot r$, then
        $\Pr_h[h(\x) = h(\y)] \leq p_2$.
    \end{itemize}
    The distribution $\mathcal H$ is called an LSH family, and has
    quality $\rho = \rho(\mathcal{H}) = \frac{\log p_1}{\log p_2}$.
\end{definition}
For the sake of simplicity, we assume that $p_2\leq 1/n$: if
$p_2>1/n$, then it suffices to create a new LSH family $\mathcal H_K$
obtained by concatenating $K=\BT{\log_{p_2}(1/n)}$ i.i.d.~hashing
functions from $\mathcal H$. The new family $\mathcal H_K$ is
$(r, cr, p_1^K, p_2^K)$-sensitive and $\rho$ does not change.

The standard approach to $(c,r)$-ANN using LSH functions is the
following.  Let $\DS$ denote the data structure constructed by \LSH,
and let $c$ denote the approximation parameter of \LSH. Each $\DS$
consists of
\begin{equation}
    \LL=n^\rho
    \eqlab{LL:value}
\end{equation}
hash functions $\ell_1,\ldots, \ell_\LL$ randomly and uniformly
selected from $\mathcal H$.  The performance of the LSH data structure
is determined by this parameter $\LL$, and one tries to minimize the
value of $\rho$ (and thus $\LL$) by picking the ``right'' hash
functions.  The data structure $\DS$ contains $\LL$ hash tables
$H_1,\ldots H_\LL$: each hash table $H_i$ contains the input set $S$
and uses the hash function $\ell_i$ to split the point set into
buckets.  For each query $\q$, we iterate over the $\LL$ hash tables:
for any hash function, compute $\ell_i(\q)$ and compute, using $H_i$,
the set
\begin{equation}
    \eqlab{bucket}
    \bucket_i(\p)=\{\p: \p\in S, \ell_i(\p)=\ell_i(\q)\}
\end{equation}
of points in $S$ with the same hash value; then, compute the distance
$\D(\q,\p)$ for each point $\p\in \bucket_i(\q)$.  The procedure stops
as soon as a $(c,r)$-near point is found. It stops and returns $\perp$
if there are no remaining points to check or if it found more than
$3L$ far points~\cite{IndykM98}.

\begin{defn}
    For a query point $\q \in \PS$, an \emphi{outlier} for an LSH data
    structure is a point $\p \in S \setminus \nbrY{\q}{c r}$, such
    that $\p \in \PS_\q = \bigcup_i \bucket_i(\q)$. An LSH
    data-structure is good if there are not too many outliers for the
    query point.  Specifically, the LSH data-structure is
    \emphi{useful} if the number of outliers for $\q$ is at most $3L$.
\end{defn}

We summarize the guarantees in the following two
lemmas~\cite{him-anntr-12}.

\begin{lemma}%
    \lemlab{nn-single:basic}%
    Consider an LSH data structure as above.  For a given query point
    $\q$, and a point $\p\in \nbrY{\q}{r}$, with probability
    $\geq 1-1/e-1/3$, we have that $\p \in \PS_\q$ and this data
    structure is useful (i.e., $\PS_\q$ contains at most $3L$ outliers
    for $\q$).
\end{lemma}

This is not quite strong enough for our purposes. We build an LSH
data-structure that uses $O( \LL \log n)$ hash functions (instead of
$\LL$).
The probability of the query point to collide with a point of
$\nbrY{\q}{r}$ is $\geq 1 - 1/n^{O(1)}$ in the new LSH structure,
while the expected number of outliers grows linearly with the number
of hash functions. We thus have the following.

\begin{lemma}%
    \lemlab{nn-single:2}%
    Consider an LSH data structure as above, using $O( \LL \log n)$ hash
    functions.  For a given query point $\q$, we have that  (i) 
     $\nbrY{\q}{r} \subseteq \PS_\q$ with probability $\geq 0.9$, 
    (ii) $\PS_\q$ contains in expectation $O(\LL \log n)$ outliers for $\q$, and
    (iii) $\PS_\q$ contains in expectation $\BO{\left( \LL + n({\q},{cr}) - n(\q,r)\right) \log n}$ points with distance larger than $r$.
\end{lemma}

The main technical issue is that before scanning $\PS_\q$, during the
query process, one can not tell whether the set $\PS_\q$ is large
because there are many points in $\nbrY{\q}{r}$, or because the data
structure contains many outliers. However, the algorithm can detect if
the LSH data structure is useless if the query process encounters too
many outliers. The following is the ``standard'' algorithm to perform
a \NN query using LSH (here the answer returned is not sampled, so
there is no fairness guarantee).

\begin{lemma}
    \lemlab{bcr:high:prob}%
    Let the query point be $\q$. Let $\DS_1, \ldots, \DS_t$ be
    $t=\Theta(\log n)$ independent LSH data structures of
    \lemref{nn-single:2}.  Then, with high probability, for a constant
    fraction of indices $j \in \IRX{t} = \{1, \ldots, t\}$, we have
    that (i) $B(\q, r) \subseteq S_{j, \q} = \bigcup_i \bucket^j_i(q)$
    and (ii) the number of outliers is
    $\cardin{S_{j,\q} \setminus \nbrY{\q}{cr}} = O( \LL \log n)$,
    where $S_{j,\q}$ is the set of all points in buckets that collide
    with $\q$.  The space used is $\dsS(n, c) = O( n \LL \log^2 n )$,
    and the expected query time is $\dsQ(n, c) =O( \LL \log n)$.
\end{lemma}

\begin{proof}
    For a random $i \in \IRX{t}$, the data structure $\DS_i$ has the
    desired properties with probability $\geq 0.9$, by
    \lemref{nn-single:2}. By Chernoff's inequality, as $t =
    \Theta(\log n)$, at least a constant fraction of these
    data-structures have the desired property.
    
    As for the query time, given a query, the data structure starts
    with $i=1$. In the $i$\th iteration, the algorithm uses from
    $\DS_i$, and computes the $O( \LL \log n)$ lists that contains the
    elements of $S_{j,\q}$. The algorithm scans these lists -- if it
    encounters more than $O( \LL \log n)$ outliers, it increases $i$
    and move on to the next data-structure. As soon as the algorithm
    encounters a near point, in these lists, it stops and returns it.
\end{proof}

\begin{remark}
    In the above, we ignored the dependency on the dimension $d$. In
    particular, the $O(\cdot)$ in hides a factor of $d$.
\end{remark}

In the following, we present data structures that solve the problems defined in \secref{problems}. 
For most of the problem variants (all except \lemref{eps:uniform:nn}) that require to return an $r$-near neighbor (but not a $cr$-near neighbor), we require that the \LSH data structure behaves well on points in $B(\q, cr) \setminus B(\q, r)$. 
Note that \defref{lsh} does not specify the behavior of the \LSH on such points.
In particular, we assume that the collision probability function of the \LSH is monotonically decreasing and that points at distance at least $r$ collide with probability $O(p_1^k)$. 
Such an \LSH has the property that the query is expected to collide with $O((n(\q, cr) - n(\q, r))\log n)$ points within distance $r$ to $cr$. 
We note that most \LSH data structures have this property naturally by providing a closed formula for the CPF of the \LSH.

\subsection{Exact Neighborhood with Dependence}
\seclab{fair:nn:dependent}

\begin{lemma}
    \lemlab{fair:nn:dependent}
    Given a set $\PS$ of $n$ points and a parameter $r$, we can
    preprocess it such that given a query $q$, one can report a point
    $p\in \PS$ with probability $1/ n(\q, r)$ 
    (points returned by subsequent queries might not be independent).
    The algorithm uses space
    $O(n\LL \log n)$ and has expected query time
    $O\left( \left(\LL + \frac{n(\q, cr)}{n(\q, r)}\right)\log^2 n
        \right)$.
\end{lemma}
\begin{proof}
      Let $\DS$ be a data structure
    constructed of \lemref{nn-single:2}.  
    Let $\Family$ be the set of all buckets in the data structure.  
    For a query point $\q$, consider the family $\FamilyA$ of all buckets containing the query, and thus $g = \cardin{\FamilyA} = O(\LL \log n)$. 
    We consider as outliers $\OL$ the set of points that are farther than $r$
    from $q$. 
    By \lemref{outlierdep}, we have that a query requires 
     $O( (\LL + \degY{\FamilyA}{\OL} / n(\q, r)) \log \LL)$ expected time (the expectation is over the random choices of the sampling data structure).
    The expected number of outliers is $\BO{(\LL +n(\q, cr)-n(\q,r)) \log n}$ (the expectation is over the random choices of LSH), since by \lemref{nn-single:2} there are 
    at most $O(\LL \log n)$ points at
    distance at least $cr$ and $\BO{(n(\q, cr)-n(\q,r)) \log n}$ points at distance between $r$ and $cr$. 
\end{proof}

The above results will always return the same point if we repeat the same query.
By using the technique in \lemref{ouliersrep}, we get the following result if the same query is repeated:

\begin{theorem}
    \thmlab{fair:nn:dependent}%
     Given a set $\PS$ of $n$ points and a parameter $r$, we can
    preprocess it such that by repeating $k$ times a query $\q$, we have with high probability $1-1/n$:
  
    \begin{enumerate}
        \item $\p$ is returned as near neighbor of $\q$ with
        probability $1/n(\q,r)$.
        \item
        $\PR{OUT_i = \p | OUT_{i-1}=\p_{i-1}, \ldots OUT_{1}=\p_{1}} =
        {1}/{n(\q,r)}$ for each $1<i\leq n$.
    \end{enumerate}
   where $OUT_i$ denotes the output of the $i$\th
    iteration with $1 \leq i \leq k$.  
    The data structure requires $\BO{n\LL \log n}$ space and
    the expected query time is
    $\BO{\left(\LL + n(\q,cr)-n(\q,r)\right) \log^2 n}$.
\end{theorem}
\begin{proof}
The claim follows by using the data structure in \lemref{ouliersrep}, where each query requires expected time  
$\BO{d_\FamilyA(\OL) \log n + (g + \delta) \log n}$,
where the $d_\FamilyA(\OL)=\BO{(\LL + n(\q,cr)-n(\q,r))\log n}$, $n$ is the total number of points,  and $\delta=g=\LL$.
\end{proof}

\subsection{Approximately Fair \ANN}
\seclab{approx:fair:nn}

\begin{theorem}%
    \thmlab{approx-neighborhood}%
    Given a set $\PS$ of $n$ points, and a parameter $r$, we can
    preprocess it in time and space
    $\dsS(n, c) = O( n \LL \log^2 n )$, see \Eqref{LL:value}. Here,
    given a query $\q$, one can report a point $\p \in \PSA$, with
    $\eps$-uniform distribution (see \defref{eps:uniform}), where
    $\PSA$ is a set that depends on the query point $\q$, and
    $\nbrY{\q}{r} \subseteq \PSA \subseteq \nbrY{\q}{cr}$.  The
    expected query time is $O( \LL \log n \log (n/\eps))$. The query
    time is $O( \LL \log^2 n \log (n/\eps))$ with high
    probability. The above guarantees hold with high probability.
\end{theorem}
\begin{proof}
    We construct the data-structures $\DS_1, \ldots, \DS_t$ of
    \lemref{bcr:high:prob}, with $t = O( \log n)$. Here, all the
    values that are mapped to a single bucket by a specific hash
    function are stored in its own set data-structure (i.e., an
    array), that supports membership queries in $O(1)$ time (by using
    hashing).

    Starting with $i=1$, the algorithm looks up the
    $\nSets = O( \LL \log n)$ buckets of the points colliding with the
    query point in the LSH data structure $\DS_i$, and let $\FamilyA$
    be this set of buckets. Let $\PSA$ be the union of the points
    stored in the buckets of $\FamilyA$. We have with constant
    probability that
    $\Cardin{\FamilyA \setminus \nbrY{\q}{cr}} \leq \mOL$, where
    $\mOL = O( \LL \log n)$. Thus, we deploy the algorithm of
    \lemref{outliers} to the sets of $\FamilyA$. With high probability
    $\nbrY{\q}{r} \subseteq \PSA$, which implies that with constant
    probability (close to $0.9$), we sample the desired point, with
    the $\eps$-uniform guarantee.  If the algorithm of
    \lemref{outliers} returns that there were too many outliers, the
    algorithm increases $i$, and try again with the next data
    structure, till success.  In expectation, the algorithm would only
    need to increase $i$ a constant number of times till success,
    implying the expected bound on the query time. With high
    probability, the number of rounds till success is $O( \log n)$.

    Since $\nSets \approx \LL = n^{\Omega(1)}$, all the high
    probability bounds here hold with probability
    $\geq 1 - 1/n^{O(1)}$.
\end{proof}

\begin{remark}
    The best value of $\,\LL$ that can be used depends on the
    underlying metric.  For the $L_1$ distance, the runtime of our
    algorithm is $\tldO(n^{1/c+o(1)})$ and for the $L_2$ distance, the
    runtime of our algorithm is $\tldO(n^{1/c^2 + o(1)})$. These match
    the runtime of the standard \LSH-based near neighbor algorithms up
    to polylog factors.
\end{remark}

\textbf{Exact neighborhood.} One can even sample $\eps$-uniformly from the $r$-near-neighbors of
the query point. Two such data structures are given in the following \lemref{eps:uniform:nn} and \remref{rem:exactANN-good-lsh} .The query time guarantee is somewhat worse.

\begin{lemma}%
    \lemlab{eps:uniform:nn}%
    Given a set $\PS$ of $n$ points, and a parameter $r$, we can
    preprocess it in time and space
    $\dsS(n, c) = O( n \LL \log^2 n )$, see \Eqref{LL:value}. Here,
    given a query $\q$, one can report a point $\p \in \nbrY{\q}{r}$,
    with $\eps$-uniform distribution (see \defref{eps:uniform}).  The
    expected query time is
    \begin{math}
        O\pth{ \LL \frac{\nNY{\q}{cr}}{\nNY{\q}{r}} \log n \log
           \frac{n}{\eps} }.
    \end{math}
    The query time is
    \begin{math}
        O\pth{ \LL \frac{\nNY{\q}{cr}}{\nNY{\q}{r}} \log^2 n \log
           \frac{n}{\eps} }
    \end{math}
    with high probability.
\end{lemma}
\begin{proof}
    Construct and use the data structure of
    \thmref{approx-neighborhood}. Given a query point, the algorithm
    repeatedly get a neighbor $\p_i$, by calling the query
    procedure. This query has a $\eps$-uniform distribution on some
    set $\PSA_i$ such that $\nbrY{\q}{r} \subseteq \PSA_i \subseteq
    \nbrY{q}{cr}$. As such, if the distance of $\p_i$ from $\q$ is
    at most $r$, then the algorithm returns it as the desired
    answer. Otherwise, the algorithm increases $i$, and continues to
    the next round.

    The probability that the algorithm succeeds in a round is
    $\rho = {\nNY{\q}{r}} / {\nNY{\q}{cr}}$, and as such the expected
    number of rounds is $1/\rho$, which readily implies the result.
\end{proof}

\begin{remark}\remlab{rem:exactANN-good-lsh}
    We remark that the properties in \lemref{eps:uniform:nn} are independent of the behavior of the \LSH with regard to points in $B(\q, cr) \setminus B(\q, r)$. 
    If the \LSH behaves well on these points as discussed in \secref{lsh}, 
    we can build the same data structure as in \thmref{approx-neighborhood} but regard all points outside of $B(\q, r)$ as outliers. 
    This results in an expected query time of 
    \begin{math}
        O\pth{ (\LL + \nNY{\q}{cr} - \nNY{\q}{r}) \log n \log
           \frac{n}{\eps} }.
    \end{math}
    With high probability the query time is 
    \begin{math}
        O\pth{ (\LL + \nNY{\q}{cr} - \nNY{\q}{r}) \log^2 n \log
           \frac{n}{\eps} }.
    \end{math}
\end{remark}

\subsection{Exact Neighborhood (Fair \NN)}
\seclab{independent:fair:nn}

We use the algorithm described in
\secref{uniform:sampling:outliers} with all points at
distance more than $r$ from the query marked as outliers.

\begin{theorem}
    \thmlab{fair:nn:independent}
    Given a set $\PS$ of $n$ points and a parameter $r$, we can
    preprocess it such that given a query $q$, one can report a point
    $p\in \PS$ with probability $1/ n(\q, r)$.  The algorithm uses space
    $O(\LL \log n)$ and has expected query time
    $O\left( \left(\LL + \frac{n(\q, cr)}{n(\q, r)}\right)
        \log^5 n\right)$.
\end{theorem}

\begin{proof}
    For $t=O(\log n)$, let $\DS_1,\ldots,\DS_t$ be data structures
    constructed by \LSH.  Let $\Family$ be the set of all buckets in
    all data structures.  For a query point $\q$, consider the family
    $\FamilyA$ of all buckets containing the query, and thus
    $g = \cardin{\FamilyA} = O(\LL \log n)$. Moreover, we let $\OL$ to
    be the set of outliers, i.e., the points that are farther than $r$
    from $\q$.  We proceed to bound the number of outliers that we
    expect to see in Step~\ref{item:outlier:uniform:reject} of the algorithm
    described in \secref{uniform:sampling:outliers}.
  
    By \lemref{nn-single:2}, we expect at most $O(g)$ points at
    distance at least $cr$.  Moreover, we set up the LSH such that the
    probability of colliding with a $cr$-near point is at most $1/g$
    in each bucket.  By the random ranks that we assign to each point,
    we expect to see $n(\q, cr)/k$ $cr$-near points in segment $h$.
    Since, $k = \Theta(n(\q, r))$, we expect to retrieve
    $O\left(g \frac{n(\q, cr)}{g n(\q, r)}\right) = O\left(\frac{n(\q,
           cr)}{n(\q, r)}\right)$ $cr$-near points in one iteration.
    With the same line of reasoning as in
    \lemref{union-of-sets-independent:outliers}, we can bound
    the expected running time of the algorithm by
    $O\left( \left(\LL + \frac{n(\q, cr)}{n(\q, r)}\right)
        \log^5 n\right)$.
\end{proof}

\begin{remark}
    With the same argument as above, we can solve Fair \NN with an
    approximate neighborhood (in which we are allowed to return points
    within distance $cr$) in expected time
    $O\bigl( \LL \log^5 n)$.
\end{remark}

We now turn our focus on designing an algorithm with a high
probability bound on its running time.  We note that we cannot use the
algorithm from above directly because with constant probability more
than $B(\q, cr) + \LL$ points are outliers.  
The following lemma, similar in nature to Lemma~\lemref{bcr:high:prob}, 
proves that by considering independent
repetitions, we can guarantee that there exists an LSH data structure
with a small number of non-near neighbors colliding with the query.

\begin{lemma}
    \lemlab{bcr:high:prob2}
    Let the query point be $\q$. Let $\DS_1, \ldots, \DS_t$ be
    $t=\Theta(\log n)$ independent LSH data structures, each
    consisting of $\Theta(\LL\log n)$ hash functions.  Then, with high
    probability, there exists $j \in \{1, \ldots, t\}$ such that (i)
    $B(\q, r) \subseteq S_{j, \q} = \bigcup_i \bucket^j_i(\q)$ and (ii)
    the number of non-near points colliding with the query (with
    duplicates) in all $\Theta(\LL \log n)$ repetitions is
    $O((n(\q, cr) + \LL) \log n)$.
\end{lemma}

\begin{proof}
    Property (i) can be achieved for all $t$ data structures
    simultaneously by setting the constant in $\Theta(\LL \log n)$ such
    that each near point has a probability of at least
    $1 - 1/(n^2\log n)$ to collide with the query.  A union bound over
    the $\log n$ independent data structure and the $n(\q, r) \leq n$
    near points then results in (i).  To see that (ii) is true,
    observe that in a single data structure, we expect not more than
    $\LL \log n$ far points and $n(\q, cr) \log n$ $cr$-near points to collide with
    the query. (Recall that each $cr$-near point is expected to collide at most once with the query.) By Markov's inequality, with probability at most $1/3$
    we see more than $3(\LL + n(\q, cr))\log n$ such points.  Using
    $t = \Theta(\log n)$ independent data structures, with high
    probability, there exists a data structure that has at most
    $3(\LL + n(\q, cr)) \log n$ points at distance larger than $r$ colliding
    with the query.
\end{proof}

\begin{lemma}
    \lemlab{fair:nn:independent:high:probability}
    Given a set $\PS$ of $n$ points and a parameter $r$, we can
    preprocess it such that given a query $q$, one can report a point
    $p\in \PS$ with probability $1/ n(\q, r)$.  The algorithm uses space
    $O(\LL \log^2 n)$ and has query time
    $O\bigl( (\LL + n(\q, cr)) \log^6 n\bigr)$ with high
    probability.
\end{lemma}

\begin{proof}
    We build $\Theta(\log n)$ independent LSH data structures, each
    containing $\Theta(\LL \log n)$ repetitions.  For a query point
    $\q$, consider the family $\FamilyA_i$ of all buckets containing
    the query in data structure $\DS_i$.  By
    \lemref{bcr:high:prob2}, in each data structure, all points
    in $B(\q, r)$ collide with the query with high probability.  We
    assume this is true for all data structures.  We change the
    algorithm presented in \secref{uniform:sampling:outliers}
    as follows:
    \begin{compactenumI}
        \item We start the algorithm by setting $k = n$.
        \item Before carrying out
        Step~\itemref{outlier:uniform:reject}, we calculate the work
        necessary to compute $\lambda_{\FamilyA_i, h}$ in each data
        structure $i \in \{1, \ldots, t\}$.  This is done by carrying
        out range queries for the ranks in segment $h$ in each bucket
        to obtain the number of points that we have to inspect.  We
        use the data structure with minimal work to carry out
        Step~\itemref{outlier:uniform:reject}.
    \end{compactenumI}

    Since all points in $B(\q, r)$ collide with the query in each data
    structure, the correctness of the algorithm follows in the same
    way as in the proof of
    \lemref{union-of-sets-independent:outliers}.

    To bound the running time, we concentrate on the iteration in
    which $k=2^{\lceil \log n(\q, r)\rceil}$, e.g.,
    $k = \Theta(n(\q,r))$.  As in the proof of
    \lemref{union-of-sets-independent:outliers}, we have to
    focus on the time it takes to compute $\lambda_{\FamilyA, h}$, the
    number of near points in all $\Theta(\LL \log n)$ repetitions in
    segment $h$ of the chosen data structure.  By
    \lemref{bcr:high:prob}, with high probability there exists
    a partition in which the number of outlier points is
    $O((n(\q, cr) + \LL) \log n)$.  Assume from here on that it exists.
    We picked the data structure with the least amount of work in
    segment $h$.  The data structure with $O((n(\q, cr) + \LL) \log n)$
    outlier points was a possible candidate, and with high probability
    there were $O((n(\q, cr) + \LL) \log^2 n)$ collisions in this data
    structure for the chosen segment.  (There could be
    $\LL\cdot n(\q, cr) \log n$ colliding near points, but with high
    probability we see at most a fraction of $\Theta(\log n/n(\q, r))$
    of those in segment $h$ because of the random ranks.  On the other
    hand, we cannot say anything about how many \emph{distinct}
    non-near points collide with the query.)  In total, we spend time
    $O((n(\q, cr) + \LL) \log^2 n)$ to compute $\lambda_{\FamilyA, h}$
    in the chosen data structure, and time $O(\LL \log^3 n)$ to find out
    which data structure has the least amount of work.  This means
    that we can bound the running time of a single iteration by
    $O((n(\q, cr) + \LL) \log^3 n)$ with high probability.  The lemma
    statement follows by observing that the segment size $k$ is
    adapted at most $\log n$ times, and for each $k$ we carry out
    $O(\log^2 n)$ rounds.
\end{proof}

\section{Fair NN using nearly-linear space LSF}
\seclab{tableau}
In this section, we describe a data structure that solves the Fair \NN
problem using the exact degree computation algorithm from
\secref{uniform}.  The bottleneck of that algorithm was
the computation of the degree of an element which takes time $O(g)$.
However, if we are to use a data structure that has at most
$O(\log n)$ repetitions, this bottleneck will be alleviated.

The following approach uses only the basic filtering approach
described in~\cite{christiani2017framework}, but no other data
structures as was necessary for solving uniform sampling with exact
probabilities in the previous sections.  It can be seen as a
particular instantiation of the more general space-time trade-off data
structures that were described
in~\cite{AndoniLRW17,christiani2017framework}.  It can also be seen as
a variant of the empirical approach discussed
in~\cite{eshghi2008locality} with a theoretical analysis.  Compared
to~\cite{AndoniLRW17,christiani2017framework}, it provides much easier
parameterization and a simpler way to make it efficient.

In this section it will be easier to state bounds on the running time
with respect to inner product similarity on unit length vectors in
$\mathbb{R}^d$.  We define the $(\alpha, \beta)$-NN problem
analogously to $(c, r)$-NN, replacing the distance thresholds $r$ and
$cr$ with $\alpha$ and $\beta$ such that $-1 < \beta < \alpha < 1$.
This means that the algorithm guarantees that if there exists a point
$\p$ with inner product at least $\alpha$ with the query point, the
data structure returns a point $\p^\ast$ with inner product at least
$\beta$ with constant probability.  The reader is reminded that for
unit length vectors we have the relation
$\norm{\p - \q}^2_2 = 2 - 2 \ip{\p}{\q}$. We will use the notation
$B_S(\q, \alpha) = \{\p \in S \mid \ip{\p}{\q} \geq \alpha\}$ and
$b_S(\q, \alpha) = |B_S(\q, \alpha)|$. We define the $\alpha$-\NNIS
problem analogously to $r$-\NNIS with respect to inner product
similarity.

We start with a succinct description of the linear space near-neighbor
data structure.  Next, we will show how to use this data structure to
solve the Fair NN problem under inner product similarity.

\subsection{Description of the data structure}

\paragraph{Construction}
Given $m \geq 1$ and $\alpha < 1$, let
$t = \lceil 1/(1 - \alpha^2)\rceil$ and assume that $m^{1/t}$ is an
integer.  First, choose $tm^{1/t}$ random vectors $\a_{i, j}$, for
$i \in [t], j \in [m^{1/t}]$, where each
$\a_{i, j} = (a_1, \ldots, a_d) \sim \mathcal{N}(0, 1)^d$ is a vector
of $d$ independent and identically distributed standard normal
Gaussians.%
\footnote{As tradition in the literature, we assume that a memory word
   suffices for reaching the desired precision. See~\cite{Charikar02}
   for a discussion.}  Next, consider a point $\p \in \PS$.  For
$i\in [t]$, let $j_i$ denote the index maximizing
$\ip{\p}{\a_{i,j_i}}$.  Then we map the index of $\p$ in $S$ to the
bucket $(j_1,\ldots,j_t) \in [m^{1/t}]^t$, and use a hash table to
keep track of all non-empty buckets.  Since a reference to each data
point in $S$ is stored exactly once, the space usage can be bounded by
$O(n + tm^{1/t})$.

\paragraph{Query} Given the query point $\q$, evaluate the dot
products with all $t m^{1/t}$ vectors $\a_{i, j}$.  For
$\varepsilon \in (0, 1)$, let
$f(\alpha, \varepsilon) = \sqrt{2(1-\alpha^2)\ln(1/\varepsilon)}$.
For $i \in [t]$, let $\Delta_{\q, i}$ be the value of the largest
inner product of $\q$ with the vectors $\a_{i, j}$ for
$j \in [m^{1/t}]$.  Furthermore, let
$\mathcal{I}_i = \{j \mid \ip{\a_{i,j}}{\q} \geq \alpha \Delta_{\q, i}
- f(\alpha, \varepsilon)\}$.  The query algorithm checks the points in
all buckets
$(i_1, \ldots, i_t) \in \mathcal{I}_1 \times \dots \times
\mathcal{I}_t$, one bucket after the other.  If a bucket contains a
close point, return it, otherwise return $\perp$.

\begin{theorem}
    Let $S \subseteq \X$ with $|S| = n$, $-1 < \beta < \alpha < 1$,
    and let $\varepsilon > 0$ be a constant.  Let
    $\rho = \frac{(1-\alpha^2)(1-\beta^2)}{(1 - \alpha\beta)^2}$.
    There exists $m = m(\alpha, \beta, n)$ such that the data
    structure described above solves the $(\alpha, \beta)$-NN problem
    with probability at least $1 - \varepsilon$ in linear space and
    expected time $n^{\rho + o(1)}$.
\end{theorem}

We remark that this result is equivalent to the running time
statements found in~\cite{christiani2017framework} for the linear
space regime, but the method is considerably simpler.  The analysis
connects storing data points in the list associated with the largest
inner product with well-known bounds on the order statistics of a
collection of standard normal variables as discussed
in~\cite{david2004order}.  The analysis is presented in
\apndref{tableau}.

\subsection[Solving a-\NNIS]{Solving $\alpha$-\NNIS}
\seclab{tableau:independent}

\paragraph*{Construction} Set $L = \Theta(\log n)$ and build $\LL$
independent data structures $\DS_1, \ldots, \DS_\LL$ as described
above. For each $\p \in \PS$, store a reference from $\p$ to the $\LL$
buckets it is stored in.

\paragraph*{Query} We run the rejection sampling approach from
\secref{uniform} on the data structure described
above. For query $\q$, evaluate all $tm^{1/t}$ filters in each
individual $\DS_\ell$.  Let $\FamilyA$ be the set of buckets
$(i_{\ell, 1},\ldots,i_{\ell,t})$ above the query threshold, for each
$\ell \in [L]$, and set $m = |\FamilyA|$.
First, check for the existence of a near neighbor by running the
standard query algorithm described above on every individual data
structure.  This takes expected time
$n^{\rho + o(1)} + O\left(\frac{b_S(\q, \beta)}{b_S(\q, \alpha) + 1}
    \log n\right)$, assuming points in a bucket appear in random
order.  If no near neighbor exists, output $\perp$ and return.
Otherwise, the algorithm performs the following steps until success is
declared:

\begin{compactenumI}
    \regVer{\smallskip}%
    \item \itemlab{linear:s:sample}%
    Picks one set from $\FamilyA$ with probabilities proportional to
    their sizes. That is, a set $\setA \in \FamilyA$ is picked with
    probability $\cardin{\setA} / m$.
    
    \item \itemlab{linear:b:sample}%
    It picks a point $\p \in \setA$ uniformly at random.
    
    \item Computes the degree $\degC = \degY{\FamilyA}{\p}$.

    \item If $\p$ is a far point, remove $\p$ from the bucket update
    the cardinality of $\setA$. Continue to the next iteration.
    
    \item If $\p$ is a near point, outputs $\p$ and stop with
    probability $1/\degC$. Otherwise, continue to the next iteration.
\end{compactenumI}

After a point $\p$ has been reported, move all far points removed
during the process into their bucket again.  As discussed in
\secref{prelims}, we assume that removing and inserting a
point takes constant time in expectation.

\begin{theorem}
    \thmlab{tableau:fair:nn}
    Let $S \subseteq \X$ with $|S| = n$ and $-1 < \beta < \alpha <
    1$. The data structure described above solves the $\alpha$-\NNIS
    problem in nearly-linear space and expected time
    $n^{\rho + o(1)} + O((b_S(\q, \beta)/(b_S(\q, \alpha) + 1)) \log^2
    n)$.
\end{theorem}

\begin{proof}
    Set $L = \Theta(\log n)$ such that with probability at least
    $1-1/n^2$, all points in $B_S(\q, \alpha)$ are found in the $T$
    buckets.  Let $\p$ be an arbitrary point in $B_S(\q, \alpha)$.
    The output is uniform by the arguments given in the proof of the
    original variant in \lemref{q:exact}.

    We proceed to prove the running time statement in the case that
    there exists a point in $B_S(\q, \alpha)$. (See the discussion
    above for the case $b_S(\q, \alpha) = 0$.)  Observe that
    evaluating all filters, checking for the existence of a near
    neighbor, removing far points, and putting far points back into
    the buckets contributes $n^{\rho + o(1)}$ to the expected running
    time (see \apndref{tableau} for details).  We did not account for
    repeatedly carrying out steps \textsf{A}--\textsf{C} yet for
    rounds in which we choose a non-far point.  To this end, we next
    find a lower bound on the probability that the algorithm declares
    success in a single such round.  First, observe that there are
    $O(b_S(\q, \beta) \log n)$ non-far points in the $T$ buckets (with
    repetitions).  Fix a point $\p \in B_S(\q, \alpha)$.  With
    probability $\Omega(c_\p/(b_S(\q, \beta) \log n))$, $\p$ is chosen
    in step B.  Thus, with probability
    $\Omega(1/(b_S(\q, \beta) \log n))$, success is declared for point
    $\p$.  Summing up probabilities over all points in
    $B_S(\q, \alpha)$, we find that the probability of declaring
    success in a single round is
    $\Omega(b_S(\q, \alpha)/(b_S(\q, \beta) \log n))$.  This means
    that we expect $O(b_S(\q, \beta)\log n/b_S(\q, \alpha))$ rounds
    until the algorithm declares success.  Each round takes time
    $O(\log n)$ for computing $c_\p$ (which can be done by marking all
    buckets that are enumerated), so we expect to spend time
    $O((b_S(\q, \beta) / b_S(\q, \alpha))\log^2 n)$ for these
    iterations, which concludes the proof.
\end{proof}

\section{Experimental Evaluation}
\seclab{evaluation}
This section presents a principled experimental evaluation that sheds
light on the general fairness implications of our problem definitions.
The aim of this evaluation is to complement the theoretical study with
a case study focusing on the fairness implications of solving variants
of the near-neighbor problem.  The evaluation contains both a
validation of the (u{}n)fairness of traditional approaches in a
recommendation setting on real-world datasets, an empirical study of
unfairness in approximate approaches, a evaluation of the average
query time of different methods in this paper, and a short discussion
of the additional cost introduced by solving the exact neighborhood
problem.  We implemented all methods and additional tools in Python 3,
and also re-implemented some special cases in C\,++ for running time
observations.  The code, raw result files, and the experimental log
containing more details are available at
\regVer{\url{https://github.com/alfahaf/fair-nn}}\dbVer{\url{https://github.com/tods21fairness/submission}}.  Moreover, the repository
contains all scripts and a Docker build script necessary to reproduce
and verify the plots presented here.

\paragraph{Datasets and Query Selection}
We run our experiments on five different datasets which are either
standard benchmarks in a recommendation system setting or in a nearest
neighbor search context (see~\cite{dataset}):
\begin{compactenumI}
    \item \MovieLens: a dataset mapping 2112 users to
    65536 unique movies.  We obtain a set representation by mapping
    each user to movies rated 4 or higher by the user, resulting in an
    average set size of 178.1 ($\sigma=187.5$).
    \item \textsf{Last.FM}: a dataset with 1892 users and 19739
    unique artists. We obtain a set representation by mapping each
    user to their top-20 artists, resulting in an average set size of
    19.8 ($\sigma=1.78$).
    \item \MNIST: a random subset of 10K points in the \MNIST training
    data set \cite{lecun1998gradient}. The full data set contains 60K
    images of hand-written digits, where each image is of size $28$ by
    $28$. Therefore, each of our points lie in a $784$ dimensional
    Euclidean space and each coordinate is in $[0,255]$.
    \item \textsf{SIFT}: We take a random subset of 10K vectors of
    the SIFT1M image descriptors that contains 1M 128-dimensional
    points.
    \item \GloVe: Finally, we take a random subset of 10K words from
    the \GloVe data set~\cite{pennington2014glove}. \si{GloVe} is a
    data set of 1.2M word embeddings in 100-dimensional space.
\end{compactenumI}
All datasets are processed automatically by our experimental
framework.  For the first two datasets, we measure the similarity of
two user sets $\mathbf{X}$ and $\mathbf{Y}$ by their Jaccard
similarity
$J(\mathbf{X},\mathbf{Y}) = |\mathbf{X} \cap \mathbf{Y}| / |\mathbf{X}
\cup \mathbf{Y}|$.  For the latter three datasets, we measure distance
by using Euclidean distance/L$_2$ norm.

For each dataset, we pick a set of ``interesting queries'' to
guarantee that the output size is not too small.  More specifically,
we consider all data points as potential queries for which the 40\th
nearest neighbor is above a certain distance threshold.  Among those
points, we choose 50 data points at random as queries and remove them
from the data set.

\paragraph{Algorithms.} 
Two different distance measures made it necessary to implement two
different LSH families.  For Jaccard similarity, we implemented LSH
using standard Min{}Hash~\cite{bro97b} and applying the 1-bit scheme
of Li and König~\cite{LiK10}.  The implementation takes two parameters
$k$ and $\LL$, as discussed in \secref{lsh}.  We choose $k$ and
$\LL$ such that the average false negative rate (the ratio of near
points not colliding with the queries) is not more than 10\%.  In
particular, $k$ is set such that we expect no more than $5$ points
with Jaccard similarity at most 0.1 to have the same hash value as the
query in a single repetition.  Both for \textsf{Last.FM} and
\MovieLens we used $k=8$ and $L=100$ with a similarity threshold of
0.2 and 0.25, respectively.

For $L_2$ Euclidean distance, we use the LSH data structure from
\cite{diim-lshsb-04}.  Each of the $\LL$ hash functions $g_i$ is the
concatenation of $k$ unit hash functions
$h_i^1\circ\cdots\circ h_i^k$.  Each of the unit hash functions
$h_i^j$ is chosen by selecting a point in a random direction (by
choosing every coordinate from a Gaussian distribution with parameters
$(0,1)$).  Then all the points are projected onto this one-dimensional
direction, and we put a randomly shifted one-dimensional grid of
length $w$ along this direction.  The cells of this grid are
considered as buckets of the unit hash function.  For tuning the
parameters of \LSH, we follow the method described in
\cite{diim-lshsb-04}, and the manual of E2LSH library \cite{E2LSH}, as
follows.

For \MNIST, the average distance of a query to its nearest neighbor in
our data set is around $1250$.  Thus, we choose the near neighbor
radius $r = 1275$.  With this choice, the $r$-neighborhood of all but
one query is non-empty.  We tune the value of $w$ and set it to
$3750$.  As before, we tune $k$ and $\LL$ so that the false negative
rate is less than $10\%$, and moreover the cost of hashing
(proportional to $\LL$) balances out the cost of scanning.  This results
in the parameter choices $k=15$ and $L=100$.  This also agrees with
the fact that $\LL$ should be roughly the square root of the total
number of points.
We use the same method for the other two data sets.  For
\textsf{SIFT}, we use $R=270$, $w=870$, $k=15$, $L=100$, and for
\textsf{GLOVE} we use $R=4.7$, $w=15.7$, $k=15$, and $L=100$.

\paragraph{Algorithms}
Given a query point $\q$, we retrieve all $\LL$ buckets corresponding to
the query.  We then implement the following algorithms and compare
their performance in returning a uniform neighbor of the query point.
\begin{compactitem}[leftmargin=0.5cm]
    \item \textbf{Uniform/Uniform}: Picks bucket uniformly at random
    and picks a random point in bucket. %
    \item \textbf{Weighted/Uniform}: Picks bucket according to its
    size, and picks uniformly random point inside bucket.
    \item \textbf{Optimal}: Picks bucket according to size, and then
    picks uniformly random point $p$ inside bucket. Then it computes
    $p$'s degree \emph{exactly} and rejects $p$ with probability
    $1-1/d(p)$.
    \item \textbf{Degree approximation}: Picks bucket according to
    size, and picks uniformly random point $p$ inside bucket. It
    approximates $p$'s degree and rejects $p$ with probability
    $1-1/d'(p)$.
    \item \textbf{Rank perturbation}: Picks the point with minimal
    rank among all buckets and perturbs the rank afterwards.
\end{compactitem}
Each method removes non-close points that might be selected from the
bucket.  We remark that the variant Uniform/Uniform most closely
resembles a standard LSH approach.  Weighted/Uniform takes the
different bucket sizes into account, but disregards the individual
frequency of a point.  Thus, the output is \emph{not expected} to be
uniform, but might be closer in distribution to the uniform
distribution.

\paragraph{Degree approximation method.} We use the algorithm of
\secref{almost:uniform} for the degree approximation: we implement a
variant of the sampling algorithm which repeatedly samples a bucket
uniformly at random and checks whether $p$ belongs to the bucket. If
the first time this happens is at iteration $i$, then it outputs the
estimate as $d'(p)=L/i$.

\paragraph{Objectives of the Experiments} Our experiments are tailored
to answer the following questions:

\begin{enumerate}
    \item[(Q1)] How (u{}n)fair is the output of different query
    algorithms in a real-world scenario?
    \item[(Q2)] What is the extra cost term for solving the exact
    neighborhood problem?
    \item[(Q3)] How quickly can the different approaches answer
    queries?
    \item[(Q4)] How fair is the output of an algorithm solving the
    approximate neighborhood version?
\end{enumerate}

\begin{figure*}[t!]
    \centering \includegraphics[width=.9\textwidth]{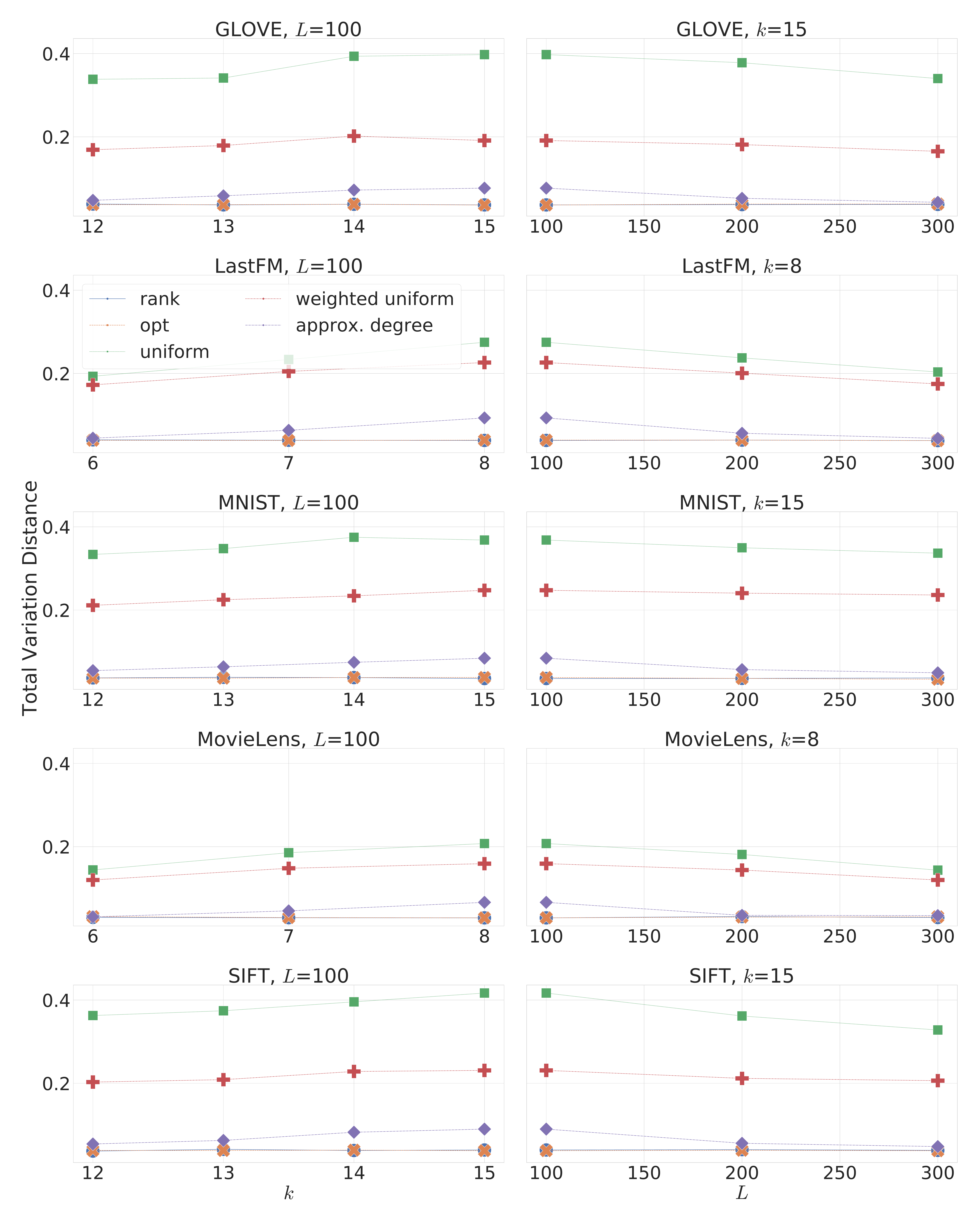}
    \caption{Total Variation Distance of the methods on the datasets.
       Parameter choices are $k=15,L=100$ for \textsf{GLOVE},
       \textsf{\MNIST}, \textsf{SIFT} and $k=8, L=100$ for
       \textsf{Last.FM} and \MovieLens.  Choices are varied to
       introduce more collisions.}
    \figlab{tvd}
\end{figure*}

\begin{figure*}[t]
    \includegraphics[width=.9\textwidth]{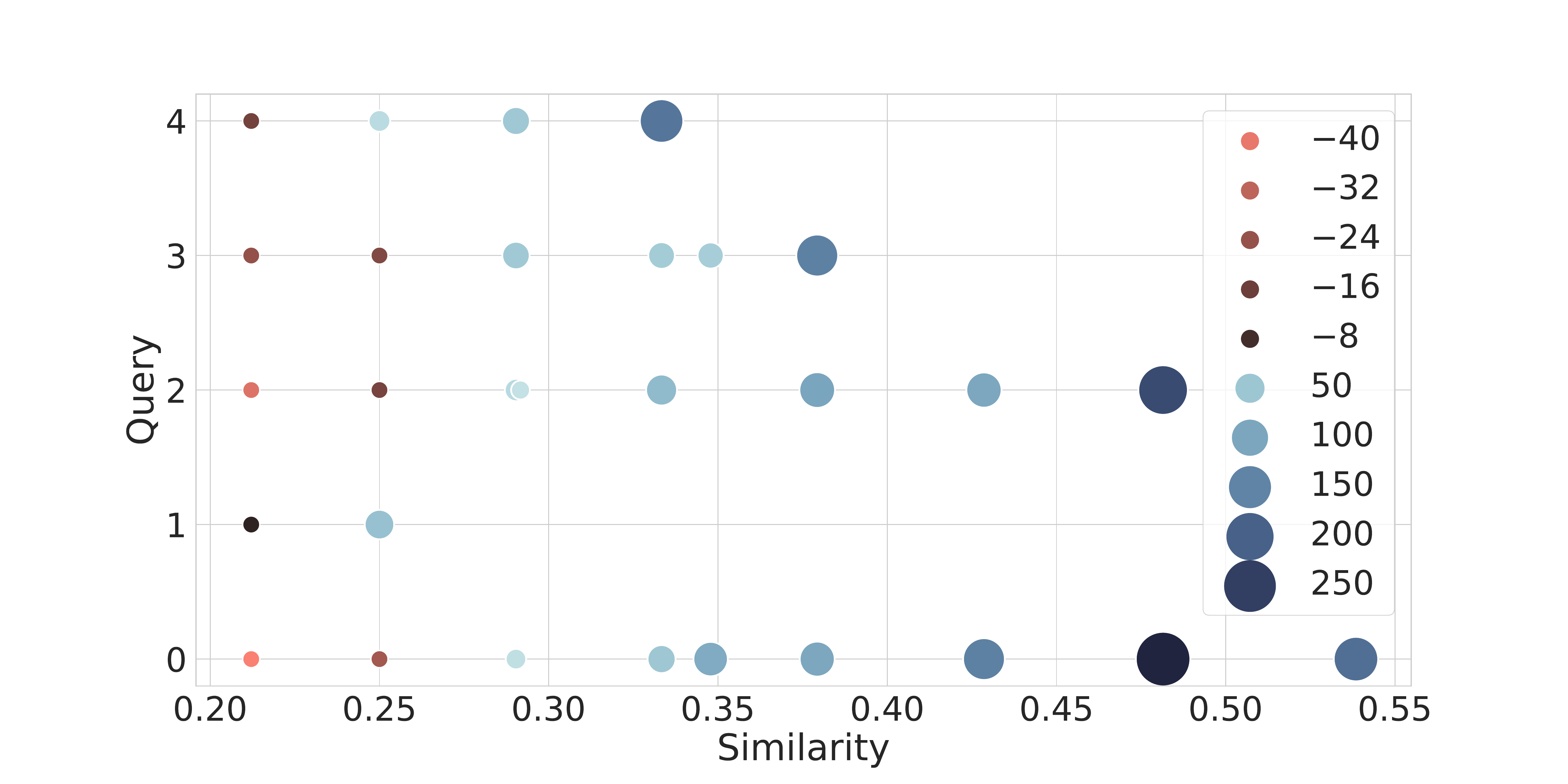}
    \caption{Scatter plot of differences in reported frequencies on
       \textsf{Last.FM} between uniform bucket (standard LSH) sampling
       and optimal (fair) sampling on five queries.  $x$-axis:
       Similarity to query point; $y$-axis: id of the query.  Each
       point represents the average difference in frequency among all
       points with the same similarity to a fixed query point.  Darker
       and larger circles means large difference in report
       frequencies.  Note large discrepancies for points at low and
       high similarity.}
    \figlab{standardfair}
\end{figure*}

\subsection{Output Distribution of Different Approaches (Q1)}

We start by building the LSH data structure with the parameters
detailed in the previous subsection.  For each query $\q$, we first
compute the number $M(\q)$ of near neighbors that collide with the
query point.  Next, we repeat the query $\q$ $100M(\q)$ times and
collect the reported points.

\figref{tvd} provides an overview over the resulting output
distribution of the five different approaches on the five datasets.
Each data point is the \emph{total variation distance}%
\footnote{For two discrete distributions $\mu$ and $\nu$ on a finite
   set $X$, the total variation distance is
   $\frac{1}{2}\sum_{x\in X} |\mu(x)-\nu(x)|$.}  (or \emph{statistical
   distance}) between the output distribution and the uniform
distribution.  Let us first concentrate on the standard choices
$k=15, L=100$ for Euclidean distance and $k=8, L=100$ for Jaccard
similarity.  We make the following observations.  Firstly, there is a
clear difference in the output distribution between the
(approximately) fair approaches and the two standard choices
\emph{uniform} and \emph{weighted uniform}.  On all datasets, there is
a large deviation from the uniform distribution with large differences
between the two approaches on the three Euclidean distance datasets,
and a marginal differences on the two Jaccard similarity datasets.
The degree approximation is much closer to the uniform distribution.
The rank approach described in
\secref{union-of-sets-dependent} is nearly indistinguishable
from the independent sampling approach that uses exact degree
computations.  This means that on the used datasets there is little
overlap between the neighborhood of the query points, and the
dependence considerations discussed in
\secref{union-of-sets-dependent} do not materialize.

As visible from the plot, we vary the parameters to introduce more
collisions between the query and dataset points in the following two
ways: We shorten the number of concatenated hash functions for a fixed
number of repetitions, and we increase the number of repetitions for a
fixed concatenation length.  In both cases, the output distributions
are closer to the uniform distribution.

Since it is difficult to argue about the unfairness of a distribution
that is around .4 away in total variation distance from the uniform
distribution, we provide a ``zoomed in'' visualization in
\figref{standardfair}.  In this plot, we show the difference
in reported frequencies between the uniform/uniform approach and fair
sampling for 5 out of the 50 queries on the \textsf{Last.FM} dataset.
Note that using our querying approach, each point should be roughly
reported 100 times.  This is true for the fair sampling approach (with
frequencies between 92 and 108), but we can clearly see a large
difference to the standard sampling approach.  This approach heavily
under-reports points at lower similarity above the threshold .6 (only
60 reports on average), and heavily over-reports points at high
similarity (360 reports in the maximum), which means that individually
there is a difference of a factor of nearly 6 in the reported
frequencies.  This means that standard approaches will indeed yield
biased neighborhoods in real-world situations.  Using algorithms that
solve the independent sampling version of the $r$-NN problem eliminate
such bias.

\subsection{Additional cost factor for solving the exact neighborhood
   variant (Q2)}
\seclab{costRatio}

The running time bounds of the algorithms in \secref{s:f:nn} and
\secref{tableau} have an additional additive or multiplicative running
time factor $\tilde{O}(b_S(\q, cr)/b_S(\q, r))$, putting in relation
the number of points at distance at most $cr, c\geq 1,$ (or similarity
at least $cr, c \leq 1,$) to those at threshold $r$.  The values $r$
and $cr$ are the distance/similarity thresholds picked when building
the data structure.  In general, a larger gap between $r$ and $cr$
makes the $n^\rho$ term in the running times smaller.  However, the
additive or multiplicative cost $\tilde{O}(b_S(\q, cr)/b_S(\q, r))$ can potentially be
prohibitively large for worst-case datasets.

\figref{ratio} depicts the ratio of points with similarity at least $cr$ and $r$ for two of the real-world datasets (the other three showed similar behavior).  
We see that one has to be careful and choose a rather small approximation factor,
since the ratio $b_S(\q, cr)/b_S(\q, r)$ can easily be the domination factor. For example, $c=2$ means that we may set $\LL$ to roughly $n^{1/4}$, which is 10. The cost of the query is then dominated by the ratio of near and $cr$-near points.
This shows that a careful inspection of dataset/query set
characteristics might be necessary to find good parameter choices that
balance cost terms.

\begin{figure}[t!]
    \includegraphics[height=4.5cm]{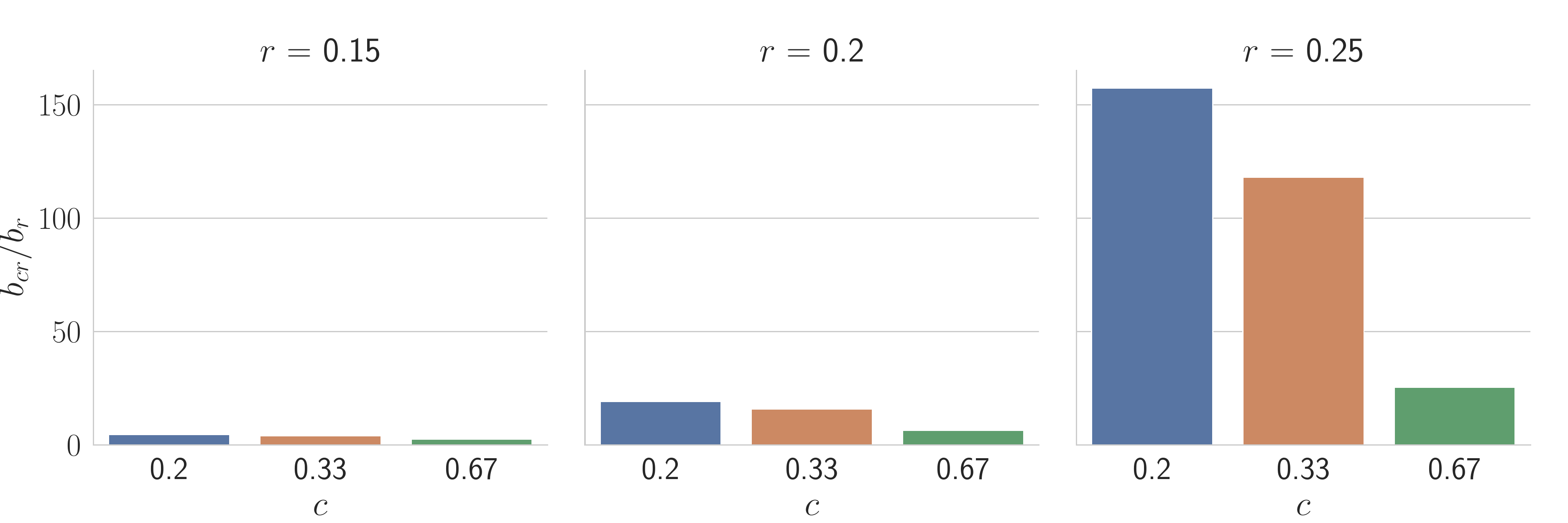}
    \includegraphics[height=4.5cm]{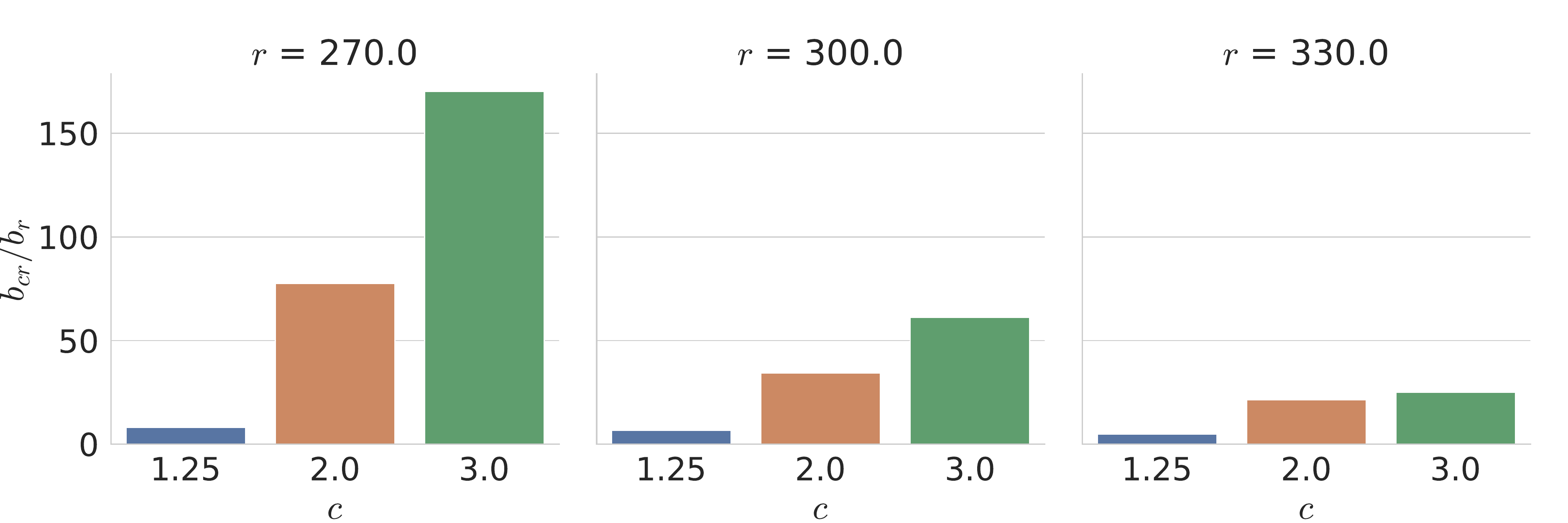}
    
    \caption{Top: Ratio of number of points with similarity at least
       $cr$ and number of points with similarity at least $r$ for
       $r\in\{0.15, 0.2, 0.25\}$ and $c = \{1/5, 1/3, 2/3\}$ for
       \MovieLens. Bottom: Ratio of number of points distance
       at most $cr$ and number of points within distance $r$ for
       $r\in\{270, 300, 300\}$ and $c = \{1.25, 2.0, 3.0\}$ for
       \textsf{SIFT}.  Note that similarity and distance behave
       inversely with regard to $r$ and $cr$.  }
    \figlab{ratio}
\end{figure}

\subsection{Running time observations (Q3)}\seclab{running_time}
We implemented LSH for Euclidean distance in C\,++ idea to compare
running times in language that generates suitable fast code.  We carry
out the same number of queries as before, but repeat each experiment
ten times.  To simplify the code structure and the use of additional
data structures, we made the following implementation choices: (i) In
the rank data structure, we store each bucket as a list of points
sorted by their rank. In the perturbation step, all buckets are sorted
again after exchanging a rank. (ii) For the other data structures, we
carry out a linear scan on the bucket sizes to select a bucket (by
weight).

\figref{running:times} reports on the average query time
needed by the different methods to answer the $100M(\q)$ queries.
These running times range from 10 to 100ms for the non-fair sampling
methods, and they are about a magnitude larger for the (approximately)
fair approaches.  With respect to exact and approximate degree
computation, we see that the former is around 1.5/2/3 times slower
than the approximate approach on \textsf{GLOVE}, \textsf{SIFT}, and
\MNIST, respectively.  The method using rank perturbation, which leads
to similar distributions as the optimal exact degree computation in
the considered workloads, is roughly as fast as the optimal variant on
\textsf{GLOVE}, 2 times faster on \textsf{SIFT}, and 3 times faster on
\MNIST.  On the \textsf{SIFT} dataset, it is roughly 2 times
faster than approximating the degree and using rejection sampling.

With a smaller number of trials, we also compared our methods to a
\naive approach that collects all colliding points, filters the near
neighbors, and returns a uniform choice in the resulting set.  This
method resulted in running times that were roughly 2-3 times slower
than the exact degree computation.

\subsection{Fairness in the approximate version
   (Q4)}\seclab{approx_fair}

We turn our focus to the approximate neighborhood sampling problem.
Recall that the algorithm may return points in
$B(\q, cr) \setminus B(\q, r)$ as well, which speeds up the query
since it avoids additional filtering steps.  In the following, we will
provide a concrete example that the output of an algorithm that solves
this problem might yield unexpected sampling results.

Let us define the following dataset over the universe
$\mathcal{U} = \{1,\ldots,30\}$: We let
$\mathbf{X} = \{16, \ldots, 30\}$, $\mathbf{Y} = \{1, \ldots, 18\}$,
and $\mathbf{Z} = \{1, \ldots, 27\}$.  Furthermore, we let
$\mathcal{M}$ contain all subsets of $\mathbf{Y}$ having at least $15$
elements (excluding $\mathbf{Y}$ itself).  The dataset is the
collection of $\mathbf{X}, \mathbf{Y}, \mathbf{Z}$ and all sets in
$\mathcal{M}$. Let the query $\mathbf{Q}$ be the set
$\{1, \ldots, 30\}$. To build the data structure, we set $r = 0.9$ and
$cr = 0.5$. It is easy to see that $\mathbf{Z}$ is the nearest
neighbor with similarity 0.9. $\mathbf{Y}$ is the second-nearest point
with similarity 0.6, but $\mathbf{X}$ is among the points with the
lowest similarity of 0.5. Finally, we note that each
$\mathbf{x}\in \mathcal{M}$ has similarity ranging from 0.5 to
$0.5\overline{6}$.  Under the approximate neighborhood sampling
problem, all points can be returned for the given query.

As in the previous subsection, the algorithm collects all the points
found in the $\LL$ buckets and returns a point picked uniformly at
random among those points having similarity at least 0.5.
\figref{difficult:sampling} shows the sampling probabilities of the
sets $\mathbf{X}, \mathbf{Y}, \mathbf{Z}$.  The plot clearly shows
that the notion of approximate neighborhood does not provide a
sensible guarantee on the individual fairness of users in this
example.  The set $\mathbf{X}$ is more than 50 times more likely than
$\mathbf{Y}$ to be returned, even though $\mathbf{Y}$ is more similar
to the query.  This is due to the clustered neighborhood of
$\mathbf{Y}$, making many other points appear at the same time in the
buckets.  On the other hand, $\mathbf{X}$ has no close points in its
neighborhood (except $\mathbf{Z}$ and $\mathbf{Q}$).

We remark that we did not observe this influence of clustered
neighborhoods on the real-world datasets. However, it is important to
notice that approximate neighborhood \emph{could} introduce
unintentional bias and, furthermore, can be exploited by an adversary
to discriminate a given user (e.g., an adversary can create a set of
objects $\mathcal{M}$ that obfuscate a given entry $\mathbf{Y}$.)

Based on this example, one could argue that the observed sampling
behavior is intentional.  If the goal is to find good representatives
in the neighborhood of the query, then it certainly seems preferable
that $\mathbf{X}$ is reported with such high probability (which is
roughly the same as all points in $\mathcal{M}$ and $\mathbf{Y}$
combined).  Our notion of fairness would make the output less diverse,
since $\mathbf{X}$ clearly stands out from many other points in the
neighborhood, but it is supposed to be treated in the very same way.
Such a trade-off between diversity and fairness was also observed, for
example, by Leonhardt \etal \cite{LeonhardtAK18}.

\begin{figure}[t]
    \begin{minipage}{.4\textwidth}
        \centering
        \includegraphics[height=5cm]{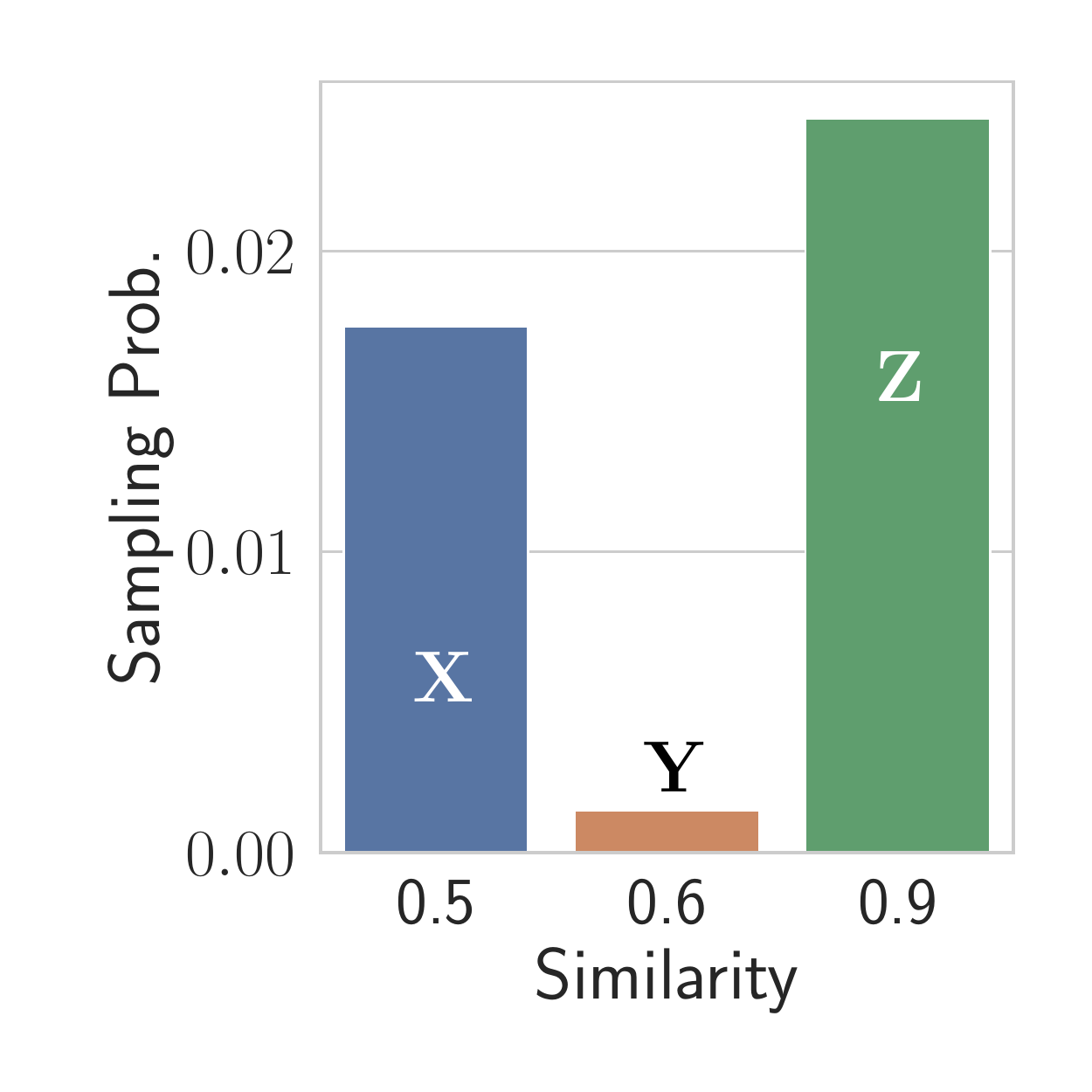}
        \caption{Empirical sampling probabilities of points
           $\mathbf{X}, \mathbf{Y}, \mathbf{Z}$.}
        \figlab{difficult:sampling}
    \end{minipage}
    \begin{minipage}{.59\textwidth}
        \centering
        \includegraphics[height=5cm]{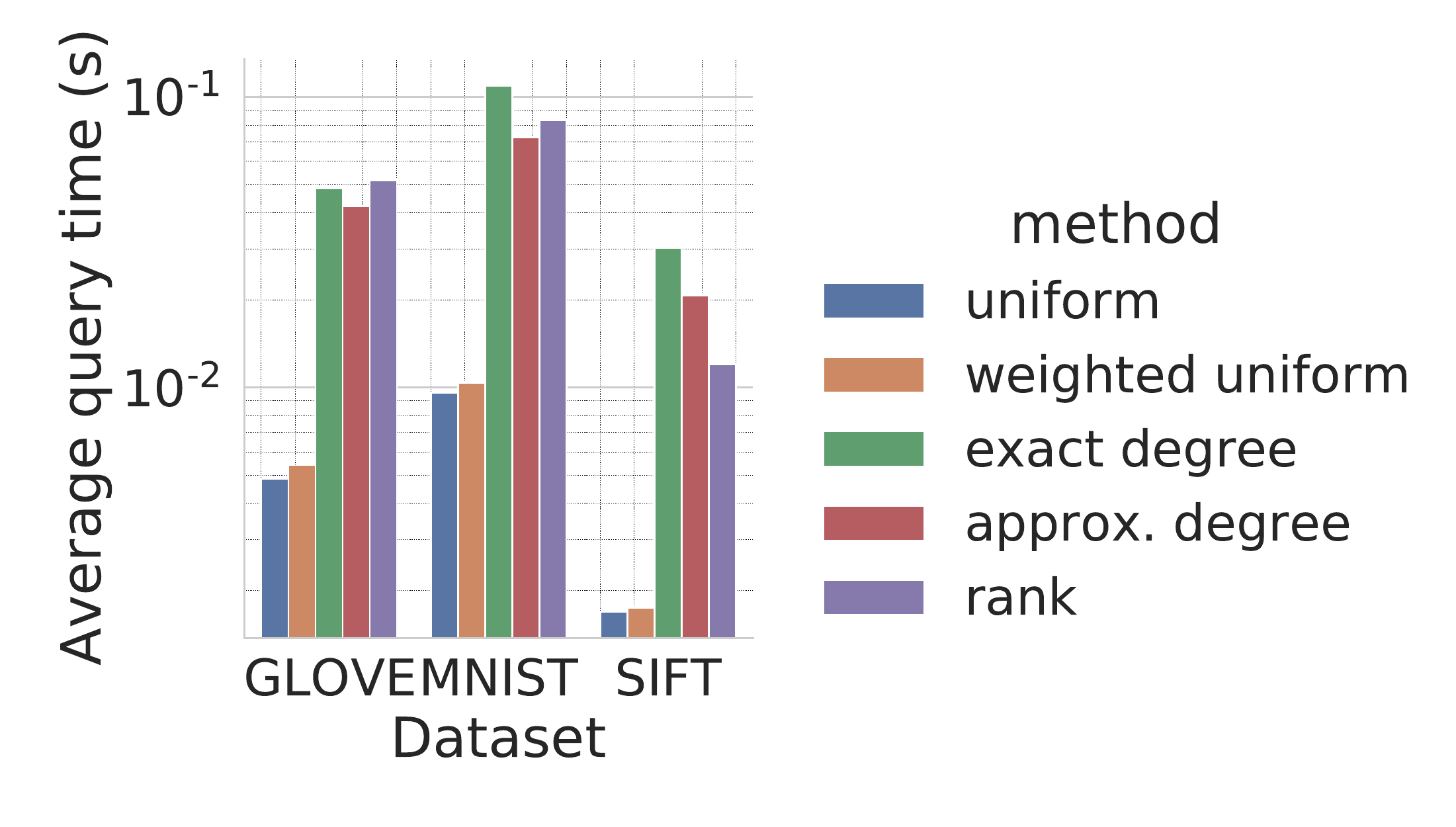}
        \caption{Comparison of average query times among the five
           methods on \textsf{GLOVE}, \MNIST, and
           \textsf{SIFT}.}
        \figlab{running:times}
    \end{minipage}
\end{figure}

\section{Conclusions}

In this paper, we have investigated a possible definition of fairness
in similarity search by connecting the notion of ``equal opportunity''
to independent range sampling.  An interesting open question is to
investigate the applicability of our data structures for problems like
discrimination discovery \cite{LuongRT11}, diversity in recommender
systems~\cite{adomavicius2014optimization}, privacy preserving
similarity search \cite{Riazi16}, and estimation of kernel density
\cite{CharikarS17}.  Moreover, it would be interesting to investigate
techniques for providing incentives (i.e., reverse
discrimination~\cite{LuongRT11}) to prevent discrimination: an idea
could be to merge the data structures in this paper with
distance-sensitive hashing functions in~\cite{Aumuller18}, which allow
to implement hashing schemes where the collision probability is an
(almost) arbitrary function of the distance.  Further, the techniques
presented here require a manual trade-off between the performance of
the LSH part and the additional running time contribution from finding
the near points among the non-far points. From a user point of view,
we would much rather prefer a parameterless version of our data
structure that finds the best trade-off with small overhead, as
discussed in~\cite{AhleAP17} in another setting. Finally, for some of the data structures presented here, the query time is also a function of the \emph{local density} around the query point (e.g. $\nNY{q}{cr}/\nNY{q}{r}$), it would be ideal to find the optimal dependence on this local density.

\regVer{
\section*{Acknowledgements}
S. Har-Peled was partially supported by a NSF AF award CCF-1907400.
R.~Pagh is part of BARC, supported by the VILLUM Foundation grant 16582.
F. Silvestri was partially supported by UniPD SID18 grant and PRIN Project n. 20174LF3T8 AHeAD. 
}

\bibliographystyle{ACM-Reference-Format}%
\bibliography{biblio}

\appendix
\section{A linear space near-neighbor data structure}
\apndlab{tableau}

We will split up the analysis of the data structure from
\secref{tableau} into two parts. First, we describe and
analyze a query algorithm that ignores the cost of storing and
evaluating the $m$ random vectors. Next, we will describe and analyze
the changes necessary to obtain an efficient query method as the one
described in \secref{tableau}.

\subsection{Description of the Data Structure}

\paragraph*{Construction} To set up the data structure for a point set
$S \subseteq \mathbb{R}^d$ of $n$ data points and two parameters
$\beta < \alpha$, choose $m \geq 1$ random vectors
$\a_1, \ldots, \a_m$ where each
$\a = (a_1, \ldots, a_d) \sim \mathcal{N}(0, 1)^d$ is a vector of $d$
independent and identically distributed standard normal Gaussians.
For each $i \in\{1,\ldots,m\}$, let $L_i$ contain all data points
$\x \in \PS$ such that $\ip{\a_i}{\x}$ is largest among all vectors
$\a$.

\paragraph*{Query} For a query point $\q \in \mathbb{R}^d$ and for a
choice of $\varepsilon \in (0, 1)$ controlling the success probability
of the query, define
$f(\alpha, \varepsilon) = \sqrt{2(1-\alpha^2)\ln(1/\varepsilon)}$.
Let $\Delta_\q$ be the largest inner product of $\q$ over all $\a$.
Let $L'_1, \ldots, L'_K$ denote the lists associated with random
vectors $\a$ satisfying
$\ip{\a}{\q} \geq \alpha \Delta_q - f(\alpha,\varepsilon)$.  Check all
points in $L'_1, \ldots, L'_K$ and report the first point $\x$ such
that $\ip{\q}{\x} \geq \beta$.  If no such point exists, report
$\perp$.

The proof of the theorem below will ignore the cost of evaluating
$\a_1, \ldots, \a_m$.  An efficient algorithm for evaluating these
vectors is provided in \apndref{tableau:efficient}.

\begin{theorem}
    \thmlab{tableau}%
    Let $-1 < \beta < \alpha < 1$, $\varepsilon \in (0, 1)$, and
    $n \geq 1$.  Let
    $\rho = \frac{(1-\alpha^2)(1-\beta^2)}{(1 - \alpha\beta)^2}$.
    There exists $m = m(n, \alpha, \beta)$ such that the data
    structure described above solves the $(\alpha, \beta)$-NN problem
    with probability at least $1 - \varepsilon$ using space $O(m + n)$
    and expected query time $n^{\rho + o(1)}$.
\end{theorem}

We split the proof up into multiple steps. First, we show that for
every choice of $m$, inspecting the lists associated with those random
vectors $\a$ such that their inner product with the query point $\q$
is at least the given query threshold guarantees to find a close point
with probability at least $1-\varepsilon$. The next step is to show
that the number of far points in these lists is $n^{\rho + o(1)}$ in
expectation.

\subsection{Analysis of Close Points}

\begin{lemma}
    \lemlab{tableau:close}
    Given $m$ and $\alpha$, let $\q$ and $\x$ such that
    $\ip{\q}{\x} = \alpha$.  Then we find $\x$ with probability at
    least $1- \varepsilon$ in the lists associated with vectors that
    have inner product at least
    $\alpha\Delta_\q - f(\alpha, \varepsilon)$ with $\q$.
\end{lemma}

\begin{proof}
    By spherical symmetry \cite{christiani2017framework}, we may
    assume that $\x = (1, 0, \ldots, 0)$ and
    $\q = (\alpha, \sqrt{1 - \alpha^2}, 0,$ $ \ldots, 0)$ . The
    probability of finding $\x$ when querying the data structure for
    $\q$ can be bounded as follows from below. Let $\Delta_\x$ be the
    largest inner product of $\x$ with vectors $\a$ and let
    $\Delta_\q$ be the largest inner product of $\q$ with these
    vectors. Given these thresholds, finding $\x$ is then equivalent
    to the statement that for the vector $\a$ with
    $\ip{\a}{\x} = \Delta_\x$ we have
    $\ip{\a}{\q} \geq \alpha \Delta_{\q} - f(\alpha,\varepsilon)$. We
    note that
    $\PR{\max\{\ip{\a}{\q}\} = \Delta} = 1 - \PR{\forall i:
       \ip{\a_i}{\q} < \Delta}$.

    Thus, we may lower bound the probability of finding $\x$ for
    arbitrary choices $\Delta_{\x}$ and $\Delta_\q$ as follows:
    \begin{align}
      \eqlab{tableau:success}
      \PR{\text{find \hspace{-0.1em} $\x$}} \hspace{-0.1em} & \hspace{-0.1em} \geq \hspace{-0.1em} \PR{\hspace{-0.1em} \ip{\a}{\q}\hspace{-0.1em} \geq \hspace{-0.1em} \alpha\Delta_{\q}{-} f(\alpha,\varepsilon) \hspace{-0.1em} \mid \hspace{-0.1em} \ip{\a}{\x} = \Delta_{\x}\hspace{-0.1em}\text{ and } \hspace{-0.1em}\ip{\a'}{\q} \hspace{-0.1em} = \hspace{-0.1em} \Delta_{\q}\hspace{-0.1em}}\notag\\ 
                                                            &\quad - \PR{\forall i: \ip{\a_i}{\x} < \Delta_\x}  - \PR{\forall i: \ip{\a_i}{\q} < \Delta_\q}.
    \end{align}
    Here, we used that
    $\PR{A \cap B \cap C} = 1 - \PR{\overline{A} \cup \overline{B}
       \cup \overline{C}} \geq \PR{A} - \PR{\overline{B}} -
    \PR{\overline{C}}$.  We will now obtain bounds for the three terms
    on the right-hand side of \Eqref{tableau:success} separately, but
    we first recall the following lemma from \cite{SZAREK1999193}:
    \begin{lemma}[\cite{SZAREK1999193}]\lemlab{normal:bound}
        Let $Z$ be a standard normal random variable. Then, for every
        $t \geq 0$, we have that
        \begin{align*}
          \frac{1}{\sqrt{2\pi}}\frac{1}{t + 1}e^{-t^2/2} \leq \Pr(Z \geq t) \leq \frac{1}{\sqrt{\pi}}\frac{1}{t + 1}e^{-t^2/2}.
        \end{align*}
    \end{lemma}

    \paragraph{Bounding the first term.} Since
    $\q = (\alpha, \sqrt{1- \alpha^2}, 0, \ldots, 0)$ and
    $\x=(1, 0, \ldots, 0)$, the condition $\ip{\a}{\x} = \Delta_\x$
    means that the first component of $\a$ is $\Delta_\x$.  Thus, we
    have to bound the probability that a standard normal random
    variable $Z$ satisfies the inequality
    $\alpha \Delta_\x + \sqrt{1 - \alpha^2} Z \geq \alpha\Delta_\q -
    f(\alpha, \varepsilon)$.  Reordering terms, we get
    \begin{align*}
      Z \geq \frac{\alpha\Delta_{\q} - f(\alpha,\varepsilon) - \alpha \Delta_\x}{\sqrt{1 - \alpha^2}}.
    \end{align*}
    Choose $\Delta_\q = \Delta_\x$.  In this case, we bound the
    probability that $Z$ is larger than a negative value.  By symmetry
    of the standard normal distribution and using
    \lemref{normal:bound}, we may compute
    \begin{align}
      \eqlab{tableau:success:2}
      \PR{Z \geq -\frac{f(\alpha, \varepsilon)}{\sqrt{1 - \alpha^2}}}
      &=1 - \PR{Z < -\frac{f(\alpha, \varepsilon)}{\sqrt{1 - \alpha^2}}}
        \notag
      \\&=1 - \PR{Z \geq \frac{f(\alpha, \varepsilon)}{\sqrt{1 - \alpha^2}}}\notag\\
      &\geq 
        1 - \frac{\text{Exp}\left(-\frac{(f(\alpha, \varepsilon))^2}{2(1 - \alpha^2)}\right)}{\sqrt{2\pi}
        \left(\frac{f(\alpha, \varepsilon)}{\sqrt{1 - \alpha^2}} + 1\right) }    \geq 1 - \varepsilon.
    \end{align}

    \paragraph{Bounding the second term and third term.} We first
    observe that
    \begin{align*}
      \PR{\forall i: \ip{\a_i}{\x} < \Delta_\x} &= \PR{\ip{\a_1}{\x} < \Delta_\x}^m \\
                                                &= \left(1 - \PR{\ip{\a_1}{\x} \geq \Delta_\x}\right)^m\\
                                                &\leq \left(1 - \frac{\text{Exp}[-\Delta_\x^2/2]}{\sqrt{2\pi}(\Delta_x + 1)}\right)^m.
    \end{align*}
    Setting $\Delta_\x = \sqrt{2 \log m - \log(4\kappa\pi \log(m))}$
    upper bounds this term by $\text{Exp}[-\sqrt{\kappa}]$.  Thus, by
    setting $\kappa \geq \log^2(1/\delta)$ the second term is upper
    bounded by $\delta \in (0, 1)$.  The same thought can be applied
    to the third summand of \Eqref{tableau:success}, which is only
    smaller because of the negative offset $f(\alpha, \varepsilon)$.

\paragraph{Putting everything together.} 
Putting the bounds obtained for all three summands together shows that
we can find $\x$ with probability at least $1 - \varepsilon'$ by
choosing $\varepsilon$ and $\delta$ such that
$\varepsilon' \geq \varepsilon + 2 \delta$.
\end{proof}

\subsection{Analysis of Far Points}

\begin{lemma}
    \lemlab{tableau:far}%
    Let $-1 < \beta < \alpha < 1$.  There exists
    $m = m(n, \alpha, \beta)$ such that the expected number of points
    $\x$ with $\ip{\x}{\q} \leq \beta$ in $L'_1,\ldots,L'_K$ where
    $K = |\{i \mid \ip{\a_i}{\q} \geq \alpha\Delta_\q - f(\alpha,
    \varepsilon)\}|$ is $n^{\rho + o(1)}$.
\end{lemma}

\begin{proof}
    We will first focus on a single far-away point $\x$ with inner
    product at most $\beta$.  Again, let $\Delta_\q$ be the largest
    inner product of $\q$. Let $\x$ be stored in $L_i$.  Then we find
    $\x$ if and only if
    $\ip{\a_i}{\q} \geq \alpha\Delta_\q - f(\alpha, \varepsilon).$ By
    spherical symmetry, we may assume that $\x = (1, 0, \ldots, 0)$
    and $\q = (\beta, \sqrt{1 - \beta^2}, 0, \ldots, 0)$.

    We first derive values $t_\q$ and $t_\x$ such that, with high
    probability, $\Delta_\q \geq t_\q$ and $\Delta_\x \leq t_\x$.
    From the proof of \lemref{tableau:close}, we know that
    \begin{align*}
      \PR{\max\{\ip{\a}{\q}\} \geq t}
      &\leq 1 - \left(1 - \frac{\text{Exp}\left(-t^2/2\right)}{\sqrt{2\pi}(t + 1)}\right)^m.
    \end{align*}

    Setting
    $t_\q = \sqrt{2\log(m / \log(n)) - \log(4 \pi \log(m / \log n))}$
    shows that with high probability we have $\Delta_\q \geq
    t_\q$. Similarly, the choice
    $t_\x = \sqrt{2\log(m \log(n)) - \log(4 \pi \log(m \log n))}$ is
    with high probability at least $\Delta_\x$. In the following, we
    condition on the event that $\Delta_\q \geq t_\q$ and
    $\Delta_\x \leq t_\x$.

    We may bound the probability of finding $\x$ as follows:
    \begin{align*}
      \PR{\ip{\a}{\q} \geq \alpha\Delta_\q - f(\alpha, \varepsilon) \mid \ip{\a}{\x} = \Delta_\x} \leq \\
      \leq \PR{\ip{\a}{\q} \geq \alpha\Delta_\q - f(\alpha, \varepsilon) \mid \ip{\a}{\x} = t_\x}\\
      \leq \PR{\ip{\a}{\q} \geq \alpha t_\q - f(\alpha, \varepsilon) \mid \ip{\a}{\x} = t_\x}.
    \end{align*}

    Given that $\ip{\a}{\x}$ is $t_\x$, the condition
    $\ip{\a}{\q} \geq \alpha t_\q - f(\alpha, \varepsilon)$ is
    equivalent to the statement that for a standard normal variable
    $Z$ we have
    $Z \geq \frac{(\alpha t_\q - f(\alpha, \varepsilon) - \beta
       t_\x)}{\sqrt{1 - \beta^2}}$.  Using
    \lemref{normal:bound}, we have
    \begin{align}
      & \PR{\ip{\a}{\q} \geq \alpha t_\q 
        \hspace{-0.1em}- \hspace{-0.1em} f(\alpha, \varepsilon) \hspace{-0.2em} \mid \hspace{-0.2em} \ip{\a}{\x} = t_\x}
        \hspace{-0.2em} \leq \hspace{-0.2em} \frac{\text{Exp}\left(-\frac{(\alpha t_\q- f(\alpha, \varepsilon)- \beta t_\x)^2}{2(1 - \beta^2)}\right)}{\sqrt{\pi} \left(\frac{(\alpha t_\q- f(\alpha, \varepsilon)- \beta t_\x)}{\sqrt{1 - \beta^2}} + 1\right)}\notag\\
      &\hspace{2em}\leq  \text{Exp}\left(-\frac{(\alpha t_\q- f(\alpha, \varepsilon)- \beta t_\x)^2}{2(1 - \beta^2)}\right)\notag\\
      &\hspace{2em}\stackrel{(1)}{=} \text{Exp}\left(-\frac{(\alpha - \beta)^2 t_\x^2}{2(1 - \beta^2)} \left(1 + O(1 / \log \log n)\right)\right)\notag\\
      &\hspace{2em}= \left(\frac{1}{m}\right)^{\frac{(\alpha - \beta)^2}{1 - \beta^2} + o(1)},
        \eqlab{far:prob}
    \end{align}
    where step (1) follows from the observation that
    $t_\q/t_\x = 1 + O(1/\log \log n)$ and
    $f(\alpha, \varepsilon)/t_\x = O(1/\log\log n)$ if
    $m = \Omega(\log n)$.

    Next, we want to balance this probability with the expected cost
    for checking all lists where the inner product with the associated
    vector $\a$ is at least
    $\alpha \Delta_\q - f(\alpha, \varepsilon)$.  By linearity of
    expectation, the expected number of lists to be checked is not
    more than
    $$m \cdot \text{Exp}\left(-(\alpha t_\q)^2\left(1/2 -
            f(\alpha,\varepsilon)/(\alpha t_\q) + f(\alpha,
            \varepsilon)^2/(2(\alpha t_\q)^2)\right)\right),$$ which
    is $m^{1-\alpha^2 + o(1)}$ using the value of $t_\q$ set above.
    This motivates to set~\Eqref{far:prob} equal to
    $m^{1-\alpha^2} / n$, taking into account that there are at most
    $n$ far-away points.  Solving for $m$, we get
    $ m = n^{\frac{1 - \beta^2}{(1 - \alpha\beta)^2} + o(1)} $ and
    this yields $m^{1-\alpha^2 +o(1)} = n^{\rho + o(1)}$.
\end{proof}

\subsection{Efficient Evaluation}
\apndlab{tableau:efficient}

The previous subsections assumed that $m$ filters can be evaluated and
stored for free.  However, this requires space and time
$n^{(1-\beta^2)/(1-\alpha\beta)^2}$, which is much higher than the
work we expect from checking the points in all filters above the
threshold. We solve this problem by using the tensoring approach,
which can be seen as a simplified version of the general approach
proposed in~\cite{christiani2017framework}.

\paragraph*{Construction} Let $t = \lceil 1/(1 - \alpha^2)\rceil$ and
assume that $m^{1/t}$ is an integer.  Consider $t$ independent data
structures $\DS_1, \ldots, \DS_t$, each using $ m^{1/t}$ random
vectors $\a_{i, j}$, for $i \in\{1,\ldots,t\}, j \in [m^{1/t}]$.  Each
$\DS_i$ is instantiated as described above.  During preprocessing,
consider each $\x \in \PS$.  If $\a_{1,i_1},\ldots,\a_{t,i_t}$ are the
random vectors that achieve the largest inner product with $\x$ in
$\DS_1, \ldots, \DS_t$, map the index of $\x$ in $S$ to the bucket
$(i_1,\ldots,i_t) \in [m^{1/t}]^t$.  Use a hash table to keep track of
all non-empty buckets.  Since each data point in $S$ is stored exactly
once, the space usage is $O(n + tm^{1/t})$.

\paragraph*{Query} Given the query point $\q$, evaluate all
$t m^{1/t}$ filters.  For $i \in \{1, \ldots, t\}$, let
$\mathcal{I}_i = \{j \mid \ip{\a_{i,j}}{\q} \geq \alpha \Delta_{\q, i}
- f(\alpha, \varepsilon)\}$ be the set of all indices of filters that
are above the individual query threshold in $\DS_i$.  Check all
buckets
$(i_1, \ldots, i_t) \in \mathcal{I}_1 \times \dots \times
\mathcal{I}_t$.  If there is a bucket containing a close point, return
it, otherwise return $\perp$.

\begin{theorem}
    Let $S \subseteq X$ with $|S| = n$ and $-1 < \beta < \alpha <
    1$. The tensoring data structure solves the $(\alpha, \beta)$-NN
    problem in linear space and expected time $n^{\rho + o(1)}$.
\end{theorem}

Before proving the theorem, we remark that efficient evaluation comes
at the price of lowering the success probability from a constant $p$
to $p^{1/(1-\alpha^2)}$.  Thus, for $\delta \in (0,1)$ repeating the
construction $\ln(1/\delta)p^{1-\alpha^2}$ times yields a success
probability of at least $1-\delta$.

\begin{proof}
    Observe that with the choice of $m$ as in the proof of
    \lemref{tableau:far}, we can bound
    $m^{1/t} = n^{(1-\alpha^2)(1-\beta^2)/(1-\alpha\beta)^2 + o(1)} =
    n^{\rho + o(1)}$.  This means that preprocessing takes time
    $n^{1+\rho + o(1)}$.  Moreover, the additional space needed for
    storing the $t m^{1/t}$ random vectors is $n^{\rho + o(1)}$ as
    well.  For a given query point $\q$, we expect that each
    $\mathcal{I}_i$ is of size $m^{(1-\alpha^2)/t + o(1)}$.  Thus, we
    expect to check not more than
    $m^{1-\alpha^2 + o(1)}=n^{\rho+o(1)}$ buckets in the hash table,
    which shows the stated claim about the expected running time.

    Let $\x$ be a point with $\ip{\q}{\x} \geq \alpha$.  The
    probability of finding $\x$ is the probability that the vector
    associated with $\x$ has inner product at least
    $\alpha\Delta_{\q,i} - f(\alpha, \varepsilon)$ in $\DS_i$, for all
    $i \in \{1, \ldots, t\}$.  This probability is $p^t$, where $p$ is
    the probability of finding $\x$ in a single data structure
    $\DS_i$.  By \thmref{tableau} and since $\alpha$ is a constant,
    this probability is constant and can be bounded from below by
    $1 - \delta$ via a proper choice of $\varepsilon$ as discussed in
    the proof of \lemref{tableau:close}.

    Let $\y$ be a point with $\ip{\q}{\y} < \beta$.  Using the same
    approach in the proof of \lemref{tableau:far}, we observe that the
    probability of finding $\y$ in an individual $\DS_i$ is
    $(1/m)^{1/t \cdot (\alpha - \beta)^2/(1-\beta^2) + o(1)}$.  Thus
    the probability of finding $\y$ in a bucket inspected for $\q$ is
    at most $(1/m)^{(\alpha - \beta)^2/(1-\beta^2) + o(1)}$. Setting
    parameters as before shows that we expect at most $n^{\rho +o(1)}$
    far points in buckets inspected for query $\q$, which completes
    the proof.
\end{proof}


\end{document}

%% file: prefix.tex
\ifx\RegVerFlag\undefined%
\newcommand{\dbVer}[1]{#1}
\newcommand{\regVer}[1]{}
\else
\newcommand{\dbVer}[1]{}
\newcommand{\regVer}[1]{#1}
\fi

\IfFileExists{sariel_computer.sty}{\def\sarielComp{1}}{}
\ifx\sarielComp\undefined%
\newcommand{\SarielComp}[1]{}
\newcommand{\NotSarielComp}[1]{#1}%
\else
\newcommand{\SarielComp}[1]{#1}%
\newcommand{\NotSarielComp}[1]{}%
\fi
\newcommand{\IfPrinterVer}[2]{#2}%

\usepackage{amsmath}%
\usepackage{xcolor}%
\usepackage{scalerel}%
\usepackage{euscript}%
\usepackage{stmaryrd}%

\SarielComp{\usepackage{sariel_colors}}%

\usepackage{mleftright}%
\usepackage{xspace}%
\usepackage{hyperref}%

\IfPrinterVer{%
   \usepackage{hyperref}%
}{%
   \usepackage{hyperref}%
   \hypersetup{%
      breaklinks,%
      ocgcolorlinks, colorlinks=true,%
      urlcolor=[rgb]{0.25,0.0,0.0},%
      linkcolor=[rgb]{0.5,0.0,0.0},%
      citecolor=[rgb]{0,0.2,0.445},%
      filecolor=[rgb]{0,0,0.4},
      anchorcolor=[rgb]={0.0,0.1,0.2}%
   }
}

\newcommand{\HLink}[2]{\hyperref[#2]{#1~\ref*{#2}}}
\newcommand{\HLinkSuffix}[3]{\hyperref[#2]{#1\ref*{#2}{#3}}}

\newcommand{\figlab}[1]{\label{fig:#1}}
\newcommand{\figref}[1]{\HLink{Figure}{fig:#1}}

\newcommand{\thmlab}[1]{{\label{theo:#1}}}
\newcommand{\thmref}[1]{\HLink{Theorem}{theo:#1}}

\newcommand{\itemlab}[1]{\label{item:#1}}
\newcommand{\itemref}[1]{\HLinkSuffix{}{item:#1}{}}
\newcommand{\stepref}[1]{\HLinkSuffix{Step}{item:#1}{}}

\providecommand{\deflab}[1]{\label{def:#1}}
\newcommand{\defref}[1]{\HLink{Definition}{def:#1}}

\newcommand{\apndlab}[1]{\label{apnd:#1}}
\newcommand{\apndref}[1]{\HLink{Appendix}{apnd:#1}}

\newcommand{\lemlab}[1]{\label{lemma:#1}}
\newcommand{\lemref}[1]{\HLink{Lemma}{lemma:#1}}%

\newcommand{\remlab}[1]{\label{remark:#1}}
\newcommand{\remref}[1]{\HLink{Remark}{remark:#1}}%

\newcommand{\seclab}[1]{\label{section:#1}}
\newcommand{\secref}[1]{\HLink{Section}{section:#1}}%

\providecommand{\eqlab}[1]{}%
\renewcommand{\eqlab}[1]{\label{equation:#1}}
\newcommand{\Eqref}[1]{\HLinkSuffix{Eq.~(}{equation:#1}{)}}

\newcommand{\remove}[1]{}%
\newcommand{\Set}[2]{\left\{ #1 \;\middle\vert\; #2 \right\}}

\newcommand{\pth}[2][\!]{\mleft({#2}\mright)}%

\newcommand{\pbrcx}[1]{\left[ {#1} \right]}%
\newcommand{\ProbLTR}{\mathbb{P}}%

\newcommand{\Prob}[1]{\mathop{\ProbLTR} \mleft[ #1 \mright]}%
\newcommand{\ProbCond}[2]{\mathop{\ProbLTR}\!\left[%
       #1 \;\middle\vert\; #2 \right]}

\newcommand{\ExChar}{\mathbb{E}}%
\newcommand{\ExSym}{\mathop{\ExChar}}%
\newcommand{\Ex}[2][\!]{\ExSym#1\pbrcx{#2}}

\newcommand{\ceil}[1]{\left\lceil {#1} \right\rceil}

\newcommand{\brc}[1]{\left\{ {#1} \right\}}
\newcommand{\cardin}[1]{\left| {#1} \right|}%

\renewcommand{\th}{th\xspace}

\renewcommand{\Re}{\mathbb{R}}%

\definecolor{blue25}{rgb}{0, 0, 11}
\dbVer{\definecolor{blue25}{rgb}{0,0,0}}%

\providecommand{\emphic}[2]{%
   \textcolor{blue25}{%
      \textbf{\emph{#1}}}%
   \index{#2}}

\ifx\PDFBlackWhite\undefined
\else
\renewcommand{\emphic}[2]{\textbf{\emph{#1}}}
\fi
\providecommand{\emphi}[1]{\emphic{#1}{#1}}
\renewcommand{\emphi}[1]{\emphic{#1}{#1}}

\providecommand{\Mh}[1]{#1}%
\usepackage{algorithm}
\usepackage{algorithmic}
\usepackage{caption}
\usepackage{subcaption}

\renewcommand{\Re}{\mathbb{R}}
\newcommand{\dist}{\Mh{\mathcal{D}}}%
\newcommand{\PS}{\Mh{S}}%
\newcommand{\PSA}{\Mh{T}}%
\newcommand{\DS}{\Mh{\EuScript{D}}}%

\newcommand{\norm}[1]{\left\|#1\right\|}%
\newcommand{\eps}{\varepsilon}

\newcommand{\LL}{\Mh{\mathsf{L}}}%

\newcommand{\Term}[1]{\textsf{#1}}
\newcommand{\TermI}[1]{\Term{#1}\index{#1@\Term{#1}}}

\newcommand{\MovieLens}{\Term{MovieLens}\xspace}%

\newcommand{\naive}{na\"ive\xspace}
\providecommand{\si}[1]{#1}

\newcommand{\GloVe}{Glo{}V{}e\xspace}

\newcommand{\LSH}{\TermI{LSH}\xspace}
\newcommand{\MNIST}{\TermI{MNIST}\xspace}%
\newcommand{\NN}{\TermI{NN}\xspace}
\newcommand{\ANN}{\TermI{ANN}\xspace}

\newcommand{\AFNN}{Approximately Fair \NN}
\newcommand{\AFANN}{Approximately Fair \ANN}
\newcommand{\prb}{\Mh{\varphi}}%
\newcommand{\tldO}{\scalerel*{\widetilde{O}}{j^2}}%

\newcommand{\dsS}{\Mh{\mathcal{S}}}%
\newcommand{\dsQ}{\Mh{\mathcal{Q}}}%

\newcommand{\bucket}{\Mh{H}}

\newcommand{\nbrY}[2]{\Mh{B_{\PS}}\pth{#1, #2}}%
\newcommand{\nNY}[2]{\Mh{n}\pth{#1, #2}}%
\newcommand{\pNY}[2]{\Mh{\ProbLTR}\pth{#1, #2}}%

\usepackage[inline]{enumitem}

\newlist{compactitem}{itemize}{5}%
\setlist[compactitem]{topsep=0pt,itemsep=-1ex,partopsep=1ex,parsep=1ex,%
   label=\ensuremath{\bullet}}%

\newlist{compactenumA}{enumerate}{5}%
\setlist[compactenumA]{topsep=0pt,itemsep=-1ex,partopsep=1ex,parsep=1ex,%
   label=(\Alph*)}%

\newlist{compactenuma}{enumerate}{5}%
\setlist[compactenuma]{topsep=0pt,itemsep=-1ex,partopsep=1ex,parsep=1ex,%
   label=(\alph*)}%

\newlist{compactenumI}{enumerate}{5}%
\setlist[compactenumI]{topsep=0pt,itemsep=-1ex,partopsep=1ex,parsep=1ex,%
   label=(\Roman*)}%

\newlist{compactenumi}{enumerate}{5}%
\setlist[compactenumi]{topsep=0pt,itemsep=-1ex,partopsep=1ex,parsep=1ex,%
   label=(\roman*)}%

\newcommand{\ground}{\Mh{\mathcal{U}}}%
\newcommand{\SC}{\Mh{B}}
\newcommand{\SA}{\ensuremath{\Mh{U}}}%
\newcommand{\epsA}{\Mh{\xi}}%

\newcommand{\BadProb}{\Mh{\gamma}}%

\newcommand{\setA}{\Mh{A}}%
\newcommand{\Family}{\Mh{\mathcal{F}}}%
\newcommand{\FamilyA}{\Mh{\mathcal{G}}}%
\newcommand{\FamilyB}{\Mh{\mathcal{H}}}%
\newcommand{\aprxEps}{\Mh{\,\ts\approx_\eps\,\ts}}%

\newcommand{\MS}{\Mh{\mathcal{X}}}%

\newcommand{\OL}{\Mh{\mathcal{O}}}%
\newcommand{\mOL}{\Mh{m^{}_\mathcal{O}}}%
\newcommand{\Tree}{\Mh{T}}%
\newcommand{\obj}{\Mh{o}}%

\newcommand{\Event}{\mathcal{E}}%

\newcommand{\ts}{\hspace{0.6pt}}%
\newcommand{\Cardin}[1]{%
   \ts\!\!\left\bracevert%
       \!\ts\!%
       \vphantom{#1}%
       #1%
       \!\!\ts%
   \right\bracevert%
   \!\!\ts%
}

\regVer{%
   \numberwithin{figure}{section}%
   \numberwithin{table}{section}%
   \numberwithin{equation}{section}%
}%

\ifx\SuppressAuthor\undefined%
\newcommand{\ShowAuthor}[1]{#1}%
\else%
\newcommand{\ShowAuthor}[1]{}
\fi

\newcommand{\degC}{\Mh{\mathsf{d}}}%
\newcommand{\DegY}[2]{\Mh{\mathsf{D}}_{#1}\pth{#2}}%
\newcommand{\degY}[2]{\Mh{\mathsf{d}}_{#1}\pth{#2}}%
\newcommand{\degX}[1]{\Mh{\mathsf{d}}\pth{#1}}%

\newcommand{\nSets}{\Mh{g}}%

\newcommand{\NNIS}{\Term{NNIS}\xspace}
\newcommand{\NNS}{\Term{NNS}\xspace}

\newcommand{\PQ}{\texttt{PQ}\xspace}
\newcommand{\etal}{\textit{et~al.}\xspace}

\newcommand{\BO}[1]{{ O}\left(#1\right)}

\newcommand{\BT}[1]{{\Theta}\left(#1\right)}

\newcommand{\PR}[1]{\textnormal{Pr}\left[#1\right]}

\newcommand{\ip}[2]{\langle{#1},{#2}\rangle}
\newcommand{\p}{\Mh{\mathbf{p}}}%
\newcommand{\q}{\Mh{\mathbf{q}}}%
\newcommand{\x}{\mathbf{x}}%
\newcommand{\y}{\mathbf{y}}%
 \renewcommand{\a}{\mathbf{a}}

\providecommand{\Mh}[1]{#1}%

\newcommand{\D}{\mathcal{D}} 
\newcommand{\X}{\mathcal{X}} 
\newtheorem{lemma}{Lemma}
\newtheorem{theorem}{Theorem} 
\newtheorem{definition}{Definition}

\newtheorem{remark}{Remark}%

\newtheorem{defn}[definition]{Definition}

\usepackage{todonotes}

\providecommand{\IntRange}[1]{\mleft\llbracket #1 \mright\rrbracket}
\newcommand{\IRX}[1]{\IntRange{#1}}%